\documentclass[reqno]{amsart}

\usepackage{amsmath,amssymb,amsthm,amsfonts}
\usepackage{color}
\usepackage{hyperref}
\numberwithin{equation}{section}

\newtheorem{theorem}{Theorem}[section]
\newtheorem{lemma}[theorem]{Lemma}
\newtheorem{prop}[theorem] {Proposition}
\newtheorem{cor}[theorem]  {Corollary}
\newtheorem{definition}[theorem] {Definition}

\theoremstyle{definition}

\theoremstyle{remark}
\newtheorem{remark}[theorem]{Remark} 
\newtheorem{example}[theorem]{Example}

\newcommand{\e}{\mathrm{e}}
\newcommand{\N}{\mathbb{N}}
\newcommand{\R}{\mathbb{R}}
\newcommand{\Z}{\mathbb{Z}}
\newcommand{\C}{\mathbb{C}}

\renewcommand{\P}{\mathbb{P}}
\newcommand{\E}{\mathbb{E}}
\newcommand{\dd}{\mathrm{d}} 
\newcommand{\eps}{\varepsilon}

\newcommand{\vect}[1]{\boldsymbol{#1}}

\newcommand{\be}{\begin{equation}}
\newcommand{\ee}{\end{equation}}
\newcommand{\ba}{\begin{equation} \begin{aligned}}
\newcommand{\ea}{\end{aligned}\end{equation}}
\newcommand{\bes}{\begin{equation*}}
\newcommand{\ees}{\end{equation*}}

\def\1{{\mathchoice {1\mskip-4mu\mathrm l}      
{1\mskip-4mu\mathrm l}
{1\mskip-4.5mu\mathrm l} {1\mskip-5mu\mathrm l}}}

\begin{document}


\title{Thermodynamics of a hierarchical mixture of cubes}
\author{Sabine Jansen}

\address{Mathematisches Institut, Ludwig-Maximilians-Universit{\"a}t, 80333 M{\"u}nchen; Munich Center for Quantum Science and Technology (MCQST), Schellingstr. 4, 80799 M{\"u}nchen, Germany.}
\email{jansen@math.lmu.de} 
\date{20 September 2019}
\begin{abstract}
We investigate a toy model for phase transitions in mixtures of incompressible droplets. The model consists of non-overlapping hypercubes in $\mathbb Z^d$ of sidelengths $2^j$, $j\in \mathbb N_0$. Cubes belong to an admissible set $\mathbb B$ such that if two cubes overlap, then one is contained in the other. Cubes of sidelength $2^j$ have activity $z_j$ and  density $\rho_j$. We prove explicit formulas for the pressure and entropy, prove a van-der-Waals type equation of state, and invert the density-activity relations. In addition we explore phase transitions for parameter-dependent activities $z_j(\mu) = \exp( 2^{dj} \mu - E_j)$. We prove a sufficient criterion for absence of phase transition, show that constant energies $E_j\equiv\lambda$ lead to a continuous phase transition, and prove a necessary and sufficient condition for the existence of a first-order phase transition. \\

\noindent\emph{Keywords}:  incompressible droplets; condensation; excluded volume; polymer partition function; hierarchical model.\\

\noindent \emph{Mathematics Subject Classification}: 82B20; 82B26.
\end{abstract}

\maketitle


\section{Introduction} 

Droplet models offer helpful guidance  for understanding  nucleation  and condensation phenomena in classical statistical physics. They are known under the header of \emph{Fisher droplet models} or \emph{Frenkel-Band theory of association equilibrium}, see~\cite{fisher1967droplet, stillinger1963frenkel-band,sator2003} and the references therein. They treat a gas of molecules as an ideal mixture of droplets of different sizes,  coming each with a partition function over internal degrees of freedom, or some approximate formula for such internal partition functions. Condensation is understood as the formation of a large droplet of macroscopic size, and explicit computations are possible under the simplifying assumption that the mixture is ideal. 

Rigorous results for droplet models that take into account excluded volume effects are sparse. Fisher proved that the phase transition for ideal droplet models subsists for a class of one-dimensional models~\cite{fisher1967droplet,fisher-felderhof1970}; the one-dimensional model serves as a counter-example to the strict convexity of the pressure as a function of interaction potentials when the class of potentials is chosen too large~\cite{fisher1972discontinuity}, compare~\cite[Chapter V.2]{israel1979book}. For particles in $\R^d$ with attractive interactions, errors in the ideal mixture approximation are bounded in~\cite{jansen-koenig2012idealmixture, jansen-koenig-metzger2015}, however the bounds do not allow for a proof of phase transitions. 

The present article proposes a toy model for which exluded volume effects and phase transitions can be understood rigorously, and that might pave the way for an application of  renormalization techniques. To motivate the model it is helpful to describe first another model that we are not yet able to treat and that connects to a joint program started in~\cite{jttu2014} and pursued in~\cite{jansen-tsagkaro2019colloids, jansen-kuna-tsagkaro2019}. Consider a mixture of hard spheres in $\R^3$. Spheres are assumed to have integer volume $k\in \N_0$ and are thought of as droplets made up of $k$ particles. 
Distinct spheres cannot overlap, and a sphere of volume $k$ comes with an energy $E_k$ that satisfies $E_k = k e_\infty +o(k)$ as $k\to \infty$ with finite  bulk energy $e_\infty$.  In order to control the distribution of sphere types it is natural to work in a multi-canonical ensemble, fixing the number $N_k$ of $k$-spheres as well as the total area $\sum_k k N_k$ covered by spheres (a substitute for the total number of particles). In the thermodynamic limit $N_k/V\to \rho_k$, $\sum_k k N_k /V \to \rho$, this results in an associated Helmholtz free energy per unit volume, which at low density should be of the form
$$
	f\bigl(\beta,(\rho_j)_{j\in \N}, \rho\bigr) = \sum_{j=1}^\infty \rho_j E_j + \rho_\infty e_\infty +  \beta^{-1}\sum_{j=1}^\infty \rho_j (\log \rho_j -1) + \text{correction terms}
$$
where $\rho_\infty: = \rho- \sum_{k=0}^\infty k \rho_k$ accounts for the possible loss of mass to very large spheres. The correction terms should capture excluded volume effects and one might hope for a convergent power series expansion in the variables $\rho_j$ and $\rho_\infty$.
The question  arises if the free energy of a given packing fraction, defined by minimizing over all compatible distributions on sphere sizes
$$
	f(\beta,\rho):= \min \Bigl\{ f\bigl(\beta,(\rho_j)_{j\in \N}, \rho\bigr)\, \Big|\, \sum_{j=1}^\infty j \rho_j \leq \rho \Bigr\},
$$
is strictly convex or has affine pieces.  For the ideal mixture the question is easily answered: If
$$
	p_c^\mathrm{ideal}(\beta):=\sum_{j=1}^\infty \exp( - \beta [E_j - j e_\infty]),\quad \rho_c^\mathrm{ideal}(\beta):= \sum_{j=1}^\infty j \exp( - \beta [E_j - j e_\infty])
$$
are both finite, then the free energy is strictly convex in $\rho < \rho_c^\mathrm{ideal}(\beta)$ and affine with slope $e_\infty$ in $\rho>\rho_\mathrm{sat}$, moreover in the latter domain the unique minimizer in the variational formula is $\rho_j = \exp( - \beta[E_j - j e_\infty])=:\rho_{j}^\mathrm{ideal}(\beta)$ and it satisfies $\rho_\infty = \rho- \sum_{j=1}^\infty j \rho_j>0$.  
At  low temperature, because of $\rho_c^\mathrm{ideal}(\beta)\to 0$ as $\beta\to \infty$, one may hope that the excluded volume effects do not destroy the existence of a first-order phase transition and that correction terms might be expressed in terms of convergent power series in the sphere size distributions $\rho_{j}^\mathrm{ideal}(\beta)$, compare Section~\ref{sec:discussion}. 

Unfortunately, currently available convergence criteria for multi-species  virial expansions~\cite{jttu2014,jansen-kuna-tsagkaro2019} impose exponential decay $\rho_j \leq  \exp( - \mathrm{const}\, j)$, which excludes the ideal equilibrium densities $\exp( - \beta [E_j - je_\infty])$. Therefore the naive argument sketched above stays somewhat speculative. The purpose of the present article is to provide an example where the argument nonetheless does work. The price we pay is a drastic simplification of the mixture of hard spheres. 
It is our impression, however, that the model is a valuable addition to rigorous results in dimension one~\cite{fisher1967droplet,jansen2015tonks}, moreover the simplification is a very natural starting point in the context of renormalization group theory~\cite{dyson1969hierarchical, brydges2007RGlectures}. 

In fact the present work was motivated by the study of a two-scale mixture of hard spheres in $\R^d$~\cite{jansen-tsagkaro2019colloids}. Integrating out the small spheres gives rise to an effective model for large spheres with new effective multi-body interactions and an effective activity, which leads to improved domains of convergence in Mayer expansions. The results from~\cite{jansen-tsagkaro2019colloids} leave open whether similar improvements can be reached in multi-scale systems, integrating out objects one by one. The present article should serve as a useful companion when trying to implement such a program. \\

\medskip \noindent
Our model consists of non-overlapping hypercubes in $\Z^d$ belonging to some admissible set $\mathbb B$. The model is a special case of a polymer system~\cite{gruber-kunz1971}.  The set $\mathbb B$ of admissible cubes is such that if two cubes overlap, then necessarily one cube is contained in the other. Concretely,
 $\mathbb B = \cup_{j=0}^\infty \mathbb B_j$ where the set $\mathbb B_j$ of $j$-blocks contains the representative cube $B_j = \{1,\ldots, 2^j\}^d$ and all its shifts by vectors $2^j\vect k$, $\vect k \in \Z^d$. Such geometries are often called \emph{hierarchical} in the context of renormalization group theory~\cite{dyson1969hierarchical,brydges2007RGlectures}. We consider both the grand-canonical ensemble and the multi-canonical ensemble. In the grand-canonical ensemble, described in detail in Section~\ref{sec:model},  $j$-blocks have activity $z_j$. In the multi-canonical ensemble we work with density variables $\rho_j$ and the overall packing fraction $\sigma$, see Section~\ref{sec:entropy}. 
  
In Section~\ref{sec:pressure} we work in the grand-canonical ensemble and  prove  explicit formulas for the pressure and block densities as functions of the activities $z_j$ (Theorems~\ref{thm:pressure} and~\ref{thm:density}). The formulas are similar to formulas for an ideal mixture, the only difference is that the activity $z_j$ is replaced with an effective activity $\widehat z_j$. The effective activity $\widehat z_j$ takes into account the volume excluded for blocks of type $k\leq j$ in the presence of a $j$-block;  it is exponentially smaller than the original activity, $\widehat z_j \leq  z_j \exp( - \mathrm{const} |B_j|)$. This feature is shared by two-scale binary mixtures or colloids~\cite{jansen-tsagkaro2019colloids}. In addition, we prove an explicit inversion formula for the activities as functions of the densities and prove an equation of state for the pressure that is a  variant of the van der Waals equation of state  (Theorem~\ref{thm:inversion}). The equations are similar to equations for discrete systems of non-overlapping rods on a line~\cite{jansen2015tonks}. 

In Section~\ref{sec:entropy} we work in the multi-canonical ensemble and prove an explicit formula for the entropy as a function of block densities $\rho_j$ and the overall packing fraction (Theorem~\ref{thm:entropy}). The entropy is the sum of the entropy of an ideal mixture plus a power series correction. The power series is absolutely convergent whenever the packing fraction is strictly smaller than $1$  (Proposition~\ref{prop:analyticity})---there is no need for exponential decay $\rho_j  \leq  \exp(- \mathrm{const} |B_j|)$.  We check that   the pressure is a Legendre transform of the entropy and compute the maximizers in the resulting variational formula for the pressure (Proposition~\ref{prop:pressure-legendre}). 

In Section~\ref{sec:phasetransition} we investigate a parameter-dependent model with activities $z_j(\mu) = \exp( \mu|B_j| - E_j)$ for some given sequence of energies $(E_j)_{j\in \N_0}$ and chemical potential $\mu \in \R$, and we investigate possible phase transitions as $\mu$ is varied. We prove a  sufficient condition for the absence of phase transitions (Theorem~\ref{thm:ptabsence}). For constant energies $E_j\equiv \lambda$ with $\lambda$ sufficiently large, the mixture of cubes has a continuous phase transition (Theorem~\ref{thm:ptcontinuous}). The proof uses a parameter-dependent fixed point iteration, and we sketch some possible connections with Mandelbrot's fractal percolation model~\cite{mandelbrot1982fractal,chayes-chayes-durrett88}. A necessary and sufficient condition for the existence of first-order phase transitions is given in Theorem~\ref{thm:firstorder}.

\section{The model}\label{sec:model}

\subsection{Lattice animals. Polymer partition function}

Fix $d\in \N$ and let $\mathbb X$ the collection of finite non-empty subsets of $\Z^d$. Elements $X$ of $\mathbb X$ are called \emph{lattice animals} or \emph{polymers}. 
For $\Lambda\subset \Z^d$ a bounded non-empty set, let 
$$
	\mathbb X_\Lambda = \{X\in \mathbb X\mid X\subset \Lambda\}. 
$$
We are interested in probability measures on finite collections of lattice animals in $\Lambda$ and define 
$$
	\Omega_\Lambda:= \Bigl\{ \omega= \{X_1,\ldots, X_r\}\, \Big|\ r\in \N_0,\  X_1,\ldots, X_r\subset \Lambda,\ \forall i \neq j:\ X_i \neq X_j\Bigr\}.
$$
The empty configuration is explicitly allowed, i.e., $\varnothing \in \Omega_\Lambda$. 
Note the one-to-one correspondence 
$$
	\Omega_\Lambda \to \{0,1\}^{\mathbb X_\Lambda},\quad \omega \mapsto \bigl( n_X(\omega)\bigr)_{X\in \mathbb X_\Lambda}
$$
given by 
$$
	n_X(\omega):= \begin{cases}
									1, &\quad X\in \omega,\\
									0, &\quad X\notin \omega.
							\end{cases} 
$$
Assume we are given a map $z:\mathbb X\to \R_+$, called \emph{activity}. For $\Lambda\subset \Z^d$ a bounded non-empty set, define 
the \emph{polymer partition function} 
$$
	\Xi_\Lambda:= 1 + \sum_{r=1}^\infty \frac{1}{r!}\ \sum_{(X_1,\ldots, X_r) \in \mathbb X_\Lambda^r} \Biggl( \prod_{i=1}^r z(X_i) \Biggr) \1_{\{\forall i \neq j:\ X_i \cap X_j = \varnothing\}}
$$
and the grand-canonical Gibbs measure, a probability  measure $\mathbb P_\Lambda$ on $\Omega_\Lambda$ given by 
$$
 	\P_\Lambda\bigl( \omega =  \{X_1,\ldots, X_r\}  \Bigr) := \frac{1}{\Xi_\Lambda}  \1_{\{\forall i \neq j:\ X_i \cap X_j = \varnothing\}} \prod_{i=1}^r z(X_i), \quad 	\P_\Lambda\bigl( \omega = \varnothing\bigr) := \frac{1}{\Xi_\Lambda}.
$$
The  probabilistically minded reader may think of $\P_\Lambda$ as independent Bernoulli variables $n_X(\omega)$ with parameters $z(X)/(1+z(X))$ conditioned on non-overlap of the polymers  $X$. 

In order to pass to the limit $\Lambda\nearrow \Z^d$ we impose conditions on the activity. 
\begin{definition} \label{def:stable} 
	For $z:\mathbb X\to \R_+$ and $\theta\in \R$, let 
	$$
		||z||_\theta:= 	\sup_{x\in \Z^d} \sum_{X\ni x} \frac{1}{|X|} z(X)\, \e^{- \theta |X|}.
	$$
	The activity $z(\cdot)$ is \emph{stable} if $||z||_
	\theta<\infty$  for some $\theta \in \R$.
\end{definition}

\noindent The definition is adapted from Gruber and Kunz~\cite[Eq. (23)]{gruber-kunz1971} who call the activity stable if instead $||z||_0<\infty$ but also observe some scaling invariance of the model~\cite[Eq. (22)]{gruber-kunz1971} see the proof of Lemma~\ref{lem:scale} below. Our definition incorporates possible rescalings into the definition of stability and allows for $\theta>0$ and activities that are exponentially large in the polymer size $|X|$. 
Stability ensures a uniform bound on the finite-volume pressure.

\begin{lemma} \label{lem:scale}
	Suppose that the activity $z(\cdot)$  is stable. Then for all $\theta \in \R$ with $||z||_\theta<\infty$ and  for all $\Lambda\subset \Z^d$, we have 
	$$
		\frac{1}{|\Lambda|} \log \Xi_\Lambda \leq \theta + \e^{-\theta} + ||z||_\theta< \infty 
	$$
\end{lemma} 

\begin{proof}
	We follow~\cite[Lemma 1]{gruber-kunz1971}. 
	Define $\Phi_\theta(X) = z(X) \exp( - \theta| X|)$ if $|X|\geq 2$ and $\Phi_\theta(\{x\}) = (1+ z(\{x\}) )\exp( - \theta)$. Then $\Xi_\Lambda$ is a sum over set partitions $\{X_1,\ldots, X_r\}$ of $\Lambda$. For example, 
	if $d=1$ and $\Lambda=\{0,1\}=B_1$, then 
	\begin{align*}
		\Xi_{\{0,1\}}& = 1+ z(\{0\})  + z(\{1\}) + z(\{0\}) z(\{1\}) + z(\{0,1\})\\
			& = \bigl(1+ z(\{0\})\bigr)\bigl(1+ z(\{1\})\bigr) + z(\{0,1\}) \\
		&= \Phi_0(\{0\}) \Phi_0(\{1\}) + \Phi_0(\{0,1\}). 		
	\end{align*}
	More generally, 
	\begin{align*}
		\Xi_\Lambda &= \sum_{\{X_1,\ldots, X_r\}} \Phi_0(X_1)\cdots \Phi_0(X_r) = \e^{|\Lambda|\theta} \sum_{\{X_1,\ldots, X_r\}} \Phi_\theta(X_1)\cdots \Phi_\theta(X_r) \\
		& = \e^{|\Lambda| \theta}\sum_{\{X_1,\ldots, X_r\}} \prod_{i=1}^r \Biggl( \sum_{x_i\in X_i} \frac{\Phi_\theta(X_i)}{|X_i|}\Biggr) \\
		&\leq \e^{|\Lambda| \theta}\Biggl( 1+ \sum_{r=1}^\infty \frac{1}{r!}\sum_{(x_1,\ldots, x_r)\in \Lambda^r} \prod_{i=1}^r \Biggl( \sum_{X_i \ni x_i} \frac{\Phi_\theta(X_i)}{|X_i|} \Biggr)\Biggr)\\
		&=  \e^{|\Lambda| \theta}\exp \Biggl( \sum_{x\in \Lambda}  \sum_{X \ni x} \frac{\Phi_\theta(X)}{|X|} \Biggr). 
	\end{align*}
	It follows that 
	\begin{equation*}
			\frac{1}{|\Lambda|} \log \Xi_\Lambda\leq (\theta + \e^{-\theta})+ \frac{1}{|\Lambda|} \sum_{x\in \Lambda} \sum_{\substack{X \in \mathbb X_\Lambda:\\  x\in X}} \frac{1}{|X|} z(X) \e^{-\theta|X|} \leq \theta + \e^{-\theta} + ||z||_\theta< \infty.  \qedhere
	\end{equation*}
\end{proof} 

\subsection{Hierarchical cubes} 

Now we specialize to activity maps $z(\cdot)$ supported on a collection $\mathbb B\subset \mathbb X$ of cubes  with the property that if $A,B \in \mathbb B$ have non-empty intersection, then necessarily $A\subset B$.  A set $B \subset \Z^d$ is called a \emph{$j$-block} if 
$$
	B= \{ k_1 2^{j} +1,\ldots, (k_1 +1) 2^j\}\times \cdots \times \{ k_d 2^{j} +1,\ldots, (k_d +1) 2^j\}
$$
for some $\boldsymbol k =(k_1,\ldots, k_d) \in \Z^d$. Let $\mathbb B_j$ be the set of $j$-blocks. The blocks $B\in \mathbb B_j$ form a tiling of $\Z^d$ consisting of the tile 
$$
	B_j:= \{1,\ldots, 2^j\}^d
$$
and non-overlapping shifts of $B_j$. Let $(z_j)_{j\in \N_0}$ be a sequence of non-negative numbers. 
We are interested in activity maps of the form 
\be \label{eq:actidef}
	z(X) = \begin{cases}
				z_j, &\quad \text{if }\ X= B\in \mathbb B_j,\\
				0, &\quad \text{if\  } X\in \mathbb X\setminus \bigcup_{j=0}^\infty \mathbb B_j.
			\end{cases} 
\ee
Thus $z_0$ is the activity of a monomer $\{x\}$ and $z_1$ the activity of a cube with sidelength $2$. 	Define
	$$
		\theta^*:= \limsup_{j\to \infty}\frac{1}{|B_j|}\log z_j.
	$$

\begin{lemma} \label{lem:stable}
	The activity~\eqref{eq:actidef} is stable if and only if $\theta^*<\infty$. 
\end{lemma} 

\begin{proof}
	For every given block type $j\in \N_0$, every point $x\in \Z^d$ belongs to exactly one $j$-block, therefore 
	$$
		||z||_\theta = \sum_{j=0}^\infty \frac{1}{|B_j|} z_j \e^{-\theta |B_j|}. 
	$$	
	If $||z||_\theta < \infty$ for some $\theta\in \R$, then $z_j \leq ||z||_\theta |B_j| \exp( \theta |B_j|)$ hence $\theta^* \leq \theta <\infty$. Conversely, if $\theta^*<\infty$, then for every $\theta>\theta^*$ we have 
	$z_j \exp( - |B_j|\theta) \leq \exp(- |B_j| (\theta - \theta^* +o(1)))$ which goes to zero exponentially fast as $j\to \infty$, therefore $||z||_\theta<\infty$ and the activity is stable. 
\end{proof} 

\subsection{Ideal mixture. Bernoulli variables} 

To help interpret subsequent formulas we recall the expression of the partition function for an ideal mixture of cubes, where cubes of different type may overlap.  For $\Lambda\in \mathbb B$, set 
$$
	\Xi_\Lambda^\mathrm{Ber} := \sum_{\omega \in \Omega_\Lambda} \prod_{X\in \omega} z(X) 
$$
with $\prod_{X\in \varnothing} z(X) =1$, and let $\P_\Lambda^\mathrm{Ber}$ be the associated probability measure on $\Omega_\Lambda$. It is straightforward to check that under $\P_\Lambda^\mathrm{Ber}$, the occupation numbers $n_X(\omega)$, $X\subset \Lambda$, are independent Bernoulli variables with 
$$
	\P_\Lambda^\mathrm{Ber} \bigl( n_X(\omega) =1\bigr) = \P_\Lambda^\mathrm{Ber}(\omega \ni X) =  \frac{z(X)}{1+z(X)}. 
$$
For the activities~\eqref{eq:actidef} and $\Lambda = \Lambda_n\in \mathbb B_n$, the finite-volume pressure of the ideal mixture is 
$$
	\frac{1}{|\Lambda|}  \log \Xi_\Lambda^\mathrm{Ber}=	\frac{1}{|\Lambda|}  \sum_{\substack{B\in \mathbb B:\\ B\subset \Lambda}} \log (1+ z(B)) =\sum_{j=0}^n \frac{1}{|B_j|} \log (1+ z_j).
$$
The infinite-volume pressure for the ideal mixture is therefore
\be \label{eq:bernoulli-pressure}
	p^\mathrm{Ber} := \lim_{\Lambda\nearrow \Z^d} 	\frac{1}{|\Lambda|}  \log \Xi_\Lambda^\mathrm{Ber} = \sum_{j=0}^\infty \frac{1}{|B_j|} \log (1+ z_j).
\ee
The factor $1/|B_j|$ reflects the lack of full translational invariance of the model: only translates by multiples of $2^j$ map a $j$-block to another admissible $j$-block.
The factor $1/|B_j|$ also appears in the relation between the expected number of $j$-blocks and the probability that a given $j$-block is present: if $B_j\subset \Lambda$ then
$$
\E_\Lambda^\mathrm{Ber} \bigl[\text{number of $j$-blocks in $\omega$}\bigr] = 	 \sum_{\substack{B\in \mathbb B_j:\\ B\subset \Lambda}} \E_\Lambda^\mathrm{Ber}\bigl[ n_B(\omega) \bigr]= \frac{|\Lambda|}{|B_j|}\,  \P^\mathrm{Ber}_\Lambda\bigl( n_{B_j}(\omega) =1\bigr).
$$

\begin{remark}[Ideal gas and Poisson variables]
	The word ``ideal mixture'' often refers to a model where not only the hard-core interaction between different types of blocks is dropped, but also the self-interaction of $j$-blocks is discarded---i.e.,  not only is the mixture ideal but in addition each component on its own is an ideal gas. The configuration space of such a system is $\N_0^{\mathbb B}$ and the occupation numbers become Poisson variables with parameters $z_j$ instead of Bernoulli variables. We have chosen the superscript ``Ber'' in order to avoid ambiguities associated with the word  ``ideal.''
\end{remark} 

\section{Pressure. Grand-canonical ensemble} \label{sec:pressure}

In the following $(\Lambda_n)_{n\in \N_0}$ represents a a growing sequence of cubes $\Lambda_n \in \mathbb B_n$ with $\Lambda_n\nearrow \Z^d$. 
The pressure in finite volume and infinite volume is 
$$
	p_n:= \frac{1}{|\Lambda_n|}\log \Xi_{\Lambda_n},\quad p:= \lim_{n\to \infty} p_n. 
$$
We  assume throughout the article that the activity is stable, i.e., $\theta^* = \limsup_{j\to \infty}\frac{1}{|B_j|}\log z_j <\infty$. 

\begin{theorem} \label{thm:pressure}
	The limit defining the pressure exists and satisfies $\theta^*\leq p<\infty$. It is expressed in terms of the effective activities 
	$$
		\widehat z_0:= z_0,\quad \widehat z_j:= z_j \e^{- |B_j| p_{j-1}}\quad(j\geq 1)
$$
	as
	$$
		p= \sum_{j=0}^\infty \frac{1}{|B_j|}\log (1+\widehat z_j). 
	$$
\end{theorem} 

\noindent Consequently the pressure for a system of non-overlapping cubes is given by a formula similar to the pressure~\eqref{eq:bernoulli-pressure} for the ideal mixture, the only difference is that the activities $z_j$ are replaced by the effective activities $\widehat z_j$. The effective activity is similar to the renormalized activity for binary mixtures from~\cite{jansen-tsagkaro2019colloids}. 

\begin{proof} 
	It is straightforward to check the recurrence relation 
	\be \label{eq:recurrence1} 
		\Xi_{\Lambda_n} = z_n + \bigl( \Xi_{\Lambda_{n-1}}\bigr)^{2^d} \qquad (n\geq 1).
	\ee
	By definition of $\widehat z_j$ and $p_j$ the recurrence relation can be rewritten as 
	$$
		\Xi_{\Lambda_n}  =  ( 1+ \widehat z_n ) \bigl( \Xi_{\Lambda_{n-1}}\bigr)^{2^d}
	$$
	which gives $p_n = p_{n-1} + \frac{1}{|\Lambda_n|} \log (1+ \widehat z_n)$. Combining with $p_0 = \log (1+ z_0 ) = \log (1+\widehat z_0)$ we find 
	\be \label{eq:fipress}
		p_n = \sum_{j=0}^n \frac{1}{|B_j|} \log (1+ \widehat z_j)
	\ee
	and the existence in $\R_+\cup \{\infty\}$ of the limit defining $p$, and its representation as an infinite series, follow. The stability of the activity guarantees that the pressure is finite, see Lemma~\ref{lem:stable}. The inequality $p\geq \theta^*$ follows from $\Xi_{\Lambda_n}\geq z_n$. 
\end{proof} 

\noindent Next we investigate the density of $j$-blocks and the packing fraction. The probability that a cube $B\subset \Lambda$ belongs to $\omega$ is 
$$
	\rho_\Lambda(B):= \P_\Lambda( \omega \ni B ) = \E_\Lambda\bigl[n_B\bigr].
$$
It depends on the type of the block only, accordingly we write $\rho_\Lambda(B) = \rho_{j,\Lambda}$ if $B\in \mathbb B_j$. The expected number of $j$-blocks per unit volume is 
\be \label{eq:nujrhoj}
	\nu_{j,\Lambda}:= \frac{1}{|\Lambda|}\sum_{\substack{B\in \mathbb B_j:\\ B\subset \Lambda}} \rho_\Lambda(B) = \frac{\rho_{j,\Lambda}}{|B_j|}.
\ee
To simplify language we refer to both $\nu_{j,\Lambda}$ and $\rho_{j,\Lambda}$ as the density of $j$-cubes, though they are strictly speaking two different objects. 
The packing fraction is the fraction of area covered by cubes 
$$
	\sigma_\Lambda := \frac{1}{|\Lambda|}\, \E_\Lambda\Bigl[ \bigl| \bigcup_{B\in \omega}B\bigr|\Bigr] = \sum_{j} |B_j| \nu_{j,\Lambda} = \sum_j \rho_{j,\Lambda}. 
$$
Below we show that the limits 
\be \label{eq:rhojdef}
	\rho_j :=\lim_{n\to \infty} \rho_{j,\Lambda_n},\quad  \sigma:= \lim_{n\to \infty}\sigma_{\Lambda_n}
\ee
exist. Notice $\sigma \leq 1$ and $\sum_{j=0}^\infty \rho_j \leq \sigma$. 

\begin{theorem} \label{thm:density}
	The limits~\eqref{eq:rhojdef} exist and satisfy the following.
	\begin{enumerate}
		\item [(a)] If $\sum_{j=0}^\infty \widehat z_j < \infty$ ,  then 
		$$
			\rho_j = \frac{\widehat z_j}{1+\widehat z_j}\prod_{k=j+1}^\infty \frac{1}{1+ \widehat z_k} >0,\quad \sigma = \sum_{j=0}^\infty \rho_j = 1 - \prod_{k=0}^\infty \frac{1}{1+\widehat z_k} <1.
		$$
		\item [(b)] If $\sum_{j=0}^\infty \widehat z_j = \infty$, then $\rho_j =0$ for all $j\in \N_0$ and $\sigma =1$, moreover $p = \theta^*
$.
	\end{enumerate} 
\end{theorem} 

\noindent Case (b) corresponds to a close-packing regime where the box $\Lambda_n$ is filled with large blocks. Case (a) corresponds to a gas of small cubes that fill only a fraction of the volume. See Section~\ref{sec:phasetransition} for examples. 

\begin{proof}
	We show first that for all $n\in \N_0$ and $j=0,\ldots, n$, we have 
	\be \label{eq:firho}
		\rho_{j,\Lambda_n} = \frac{\widehat z_j}{1+\widehat z_j} \frac{1}{1+\widehat z_{j+1}}\cdots \frac{1}{1+\widehat z_n}, \quad \sigma_{\Lambda_n} = 1- \prod_{j=0}^n \frac{1}{1+\widehat z_j}.
	\ee
	The proof of the first part of~\eqref{eq:firho} is by induction over $n\geq j$ at fixed $j\in \N_0$. If $n=j$, then 
	$$
		\rho_{j,\Lambda_j} =  \P_{\Lambda_j}(\omega = \{B_j\}) = \frac{z_j}{\Xi_{\Lambda_j}} = \frac{z_j}{(1+\widehat z_j) \Xi_{\Lambda_{j-1}}^{2^d}} = \frac{\widehat z_j}{1+\widehat z_j}. 
	$$
	For the induction step, write $\Lambda_n$ as a disjoint union of $2^d$ cubes $\Lambda_{n-1}^{(k)} \in \mathbb B_{n-1}$. Let 
	$$
		\omega_k:= \{ B \in \omega \mid B\subset \Lambda_{n-1}^{(k)}\}
	$$
	so that $\omega =\omega_1\cup \cdots\cup \omega_{2^d}$, unless $\omega = \{\Lambda_n\}$ contains an $n$-block. Conditional  on $\Lambda_n\notin \omega$, the projections $\omega_1,\ldots,\omega_{2^d}$ are independent, their distribution is given by the Gibbs measures $\P_{\Lambda_{n-1}^{(k)}}$, $k=1,\ldots, 2^d$. Thus fixing a $j$-block $B\subset \Lambda_{n}$, and assuming without loss of generality $B\subset \Lambda_{n-1}^{(1)}$, we get
	\begin{align*}
		\P_{\Lambda_n}(B\in \omega) & = \P_{\Lambda_n}(B \in \omega_{\Lambda_{n-1}} \mid \Lambda_n\notin \omega) \times \P_{\Lambda_n}(\Lambda_n \notin \omega)  \\
			& = \P_{\Lambda_{n-1}^{(1)}} (B\in \omega_1) \times \frac{1}{1+\widehat z_n} 
			 = \Biggl( \frac{\widehat z_j}{1+\widehat z_j} \prod_{k=j+1}^{n-1}\frac{1}{1+\widehat z_k}\Biggr) \frac{1}{1+\widehat z_n} 
	\end{align*} 
	which is precisely the first part  of~\eqref{eq:firho}. Thus the induction step is complete. For the second part of~\eqref{eq:firho}, set $x_j= \widehat z_j / (1+ \widehat z_j)$ and $y_j = 1 - x_j$. Then 
	$$
		1 = \prod_{j=0}^n (x_j + y_j) = x_n + y_n \prod_{j=0}^{n-1} (x_j + y_j) = x_n + y_n x_{n-1}+\cdots + y_n \cdots y_1 x_0 + y_n\cdots y_0
	$$
	hence 
	$$
		1- \prod_{j=0}^n y_j = \sum_{j=0}^n x_j y_{j+1}\cdots y_n
	$$
	which is the second part of~\eqref{eq:firho}. 
	
	If $\sum_{j=0}^\infty \widehat z_j<\infty$, then the infinite product $\prod_{j=0}^\infty (1+ \widehat z_j)^{-1}$ is strictly smaller than $1$ (because the logarithm is finite). 
	We pass to the limit in~\eqref{eq:firho} and obtain part (a) of the theorem. 
	
	If $\sum_{j=0}^\infty \widehat z_j =\infty$, then $\sum_{j=0}^\infty \log (1+\widehat z_j) =\infty$ and $\lim_{n\to \infty}\prod_{j=0}^n (1+\widehat z_j)^{-1} =1$. Passing to the limit in~\eqref{eq:firho} we see that $\rho_j=0$ for all $j\in \N_0$ and $\sigma =1$. It remains to check that $p = \theta^*$. 
	We already know by Theorem~\ref{thm:pressure} that $p\geq \theta^*$. 
		Suppose by contradiction that $p>\theta^*$. 
	In view of  $p = \sum_{j} \frac{1}{|B_j|} \log (1+ \widehat z_j)< \infty$ we have $\widehat z_j \leq \exp( |B_j| p)$. If $p> \theta^*$, then we would deduce that 
	$$
		\sum_{j=0}^\infty  \widehat z_j = \sum_{j=0}^\infty  z_j \e^{- |B_j| (p+o(1))}
		\leq \sum_{j=0}^\infty \e^{- |B_j|( p - \theta^*+o(1))} <\infty,
	$$
	contradicting the assumption $\sum \widehat z_j =\infty$. Thus $p \leq \theta^*$ and $p = \theta^*$.
\end{proof} 

\noindent Next we turn to the equation of state and the inversion of the density-activity relation in the gas phase. 

\begin{theorem} \label{thm:inversion}
	Assume $\sum_{j=0}^\infty \widehat z_j < \infty$. Then 
	\be \label{eq:vanderwaals}
		p = \sum_{j=0}^\infty \frac{1}{|B_j|} \log \Bigl(1 + \frac{\rho_j}{1- \sum_{k=j}^\infty \rho_k}\Bigr)
	\ee
	and for all $j\in \N_0$
	$$
		z_j = \frac{\rho_j \exp( |B_j| p_{j-1})}{1 - \sum_{k=j}^\infty \rho_k}, \quad p_{j-1} = \sum_{k=0}^{j-1} \frac{1}{|B_k|} \log \Bigl( 1+ \frac{\rho_k}{1 - \sum_{\ell=k}^\infty \rho_\ell} \Bigr).
	$$
	with the convention $p_{-1}=0$. 
\end{theorem} 

\noindent The equations are  strikingly similar to the formulas for a one-dimensional system of non-overlapping rods~\cite[Theorem 2.12]{jansen2015tonks}. The equation of state~\eqref{eq:vanderwaals} is a variant of the van-der-Waals equation of state.

\begin{proof}
	We show first that for all
	$n\in \N$ and $j\in \{0,\ldots, n\}$, 
	\be \label{eq:fized}
		\widehat z_j = \frac{\rho_{j,\Lambda_n}}{ 1 - \sum_{k=j}^n \rho_{k,\Lambda_n}},\qquad 
		\alpha_{j,\Lambda_n}:= \prod_{k=j}^n \frac{1}{1+\widehat z_j}= 1- \sum_{k=j}^n \rho_{k,\Lambda_n}. 
	\ee
	The proof is over a finite backward induction over $j\leq n$ at fixed $n$. 	For $j=n$, we have $\rho_{n,\Lambda_n} = \widehat z_n / (1+\widehat z_n)$ by~\eqref{eq:firho} hence $\widehat z_n = \rho_{n,\Lambda_n} / (1- \rho_{n,\Lambda_n})$. Furthermore, $(1+ \widehat z_n)^{-1} = 1- \rho_{n,\Lambda_n}$. For the induction step, note 
	$$
		\rho_{j,\Lambda_n} = \frac{\widehat z_j}{1+ \widehat z_j}  \prod_{k=j+1}^n \frac{1}{1+ \widehat z_k} = \frac{\widehat z_j}{1+ \widehat z_j} \, \alpha_{j+1,\Lambda_n}.
	$$
	It follows that 
	$$
		\widehat z_j = \frac{\rho_{j,\Lambda_n}}{\alpha_{j+1,\Lambda_n}- \rho_{j,\Lambda_n}} = \frac{\rho_{j,\Lambda_n}}{1- \sum_{k=j}^n \rho_{j,\Lambda_n}}
	$$
	and 
	$$
		\alpha_{j,\Lambda_n} = \frac{1}{1+\widehat z_j} \, \alpha_{j+1,\Lambda_n}= \Bigl( 1- \frac{\rho_{j,\Lambda_n}}{\alpha_{j+1,\Lambda_n}}\Bigr) \alpha_{j+1,\Lambda_n} = 1 - \sum_{k=j}^n  \rho_{k,\Lambda_n}. 
	$$
	The induction step is complete. 
	
	If  $\sum_{j=1}^\infty \widehat z_j<\infty$, then we may pass to the limit $n\to \infty$ in~\eqref{eq:fized} with the help of Theorem~\ref{thm:density}(a) and find 
	$$
		 \widehat z_j = \frac{\rho_j}{1 - \sum_{k=j}^\infty \rho_k}.
	$$	
	Theorem~\ref{thm:pressure} and Eq.~\eqref{eq:fipress} in the proof of the theorem yield the formulas for $p$ and $p_n$, the expression for $z_j$ follows as well. 
\end{proof} 

\section{Entropy. Multi-canonical ensemble} \label{sec:entropy}

\subsection{Explicit formula. Effective densities}

Here we compute the entropy in a multi-canonical ensemble, fixing the number of $j$-blocks for each $j$. For $\omega \in \Omega$, let $N_j(\omega)$ be the number of $j$-blocks in $\omega$.  For $n\in \N$, $\Lambda_n\in \mathbb B_n$, and $N_0^{(n)},\ldots, N_n^{(n)} \in \N_0$, let 
$$
	S_{\Lambda_n} (N_0^{(n)},\ldots, N_n^{(n)}) = \log \bigl|\{\omega \in \Omega_\Lambda\mid \forall j:\,  N_j(\omega) = N_j^{(n)} \}\bigr|.
$$
Set
\be \label{eq:entropydef}
	s\bigl( (\rho_j)_{j\in \N_0},\sigma\bigr):= \lim_{n\to \infty} \frac{1}{|\Lambda_n|} \log S_{\Lambda_n} (N_0^{(n)},\ldots, N_n^{(n)}) 
\ee
where the limit is taken along sequences such that $\sum_{j=0}^n |B_j|\, N_j^{(n)}\leq |\Lambda_n|$ and
\be \label{eq:Nconv}
	\frac{1}{|\Lambda_n|} \sum_{j=0}^n |B_j|\, N_j^{(n)}\to \sigma,\quad \forall j\in \N_0:\, \frac{N_j^{(n)}}{|\Lambda_n|}\to \frac{\rho_j}{|B_j|}. 
\ee
Notice that if~\eqref{eq:Nconv} holds true, then necessarily 
$$
	\sum_{j=0}^\infty \rho_j = \sum_{j=0}^\infty \lim_{n\to \infty}\frac{|B_j|\,N_j^{(n)}}{|\Lambda_n|} \leq \lim_{n\to \infty} \sum_{j=0}^\infty \frac{|B_j|\,N_j^{(n)}}{|\Lambda_n|} = \sigma.
$$
In the sequel it is convenient to introduce, given  $(\rho_j)_{j\in \N_0}$ and $\sigma\geq \sum_{k=0}^\infty \rho_j$, the variables 
\be \label{eq:sigmajdef}
	\sigma_\infty:= \sigma - \sum_{k=0}^\infty \rho_k,\quad \sigma_j := \sigma - \sum_{k=0}^{j-1} \rho_k = \sigma_\infty + \sum_{k=j}^\infty \rho_j. 
\ee
The variable $\sigma_\infty$ represents, roughly, the fraction of volume covered by blocks that grow with $n$, while $\sigma_j$ is the fraction of volume covered by blocks of type $k\geq j$. Note that if $\sigma = \sigma_\infty + \sum_{j=0}^\infty \rho_j \leq 1$, then $\rho_j \leq 1 - \sigma_{j+1}$ for all $j\in \N_0$. 

\begin{theorem} \label{thm:entropy}
	Let $\vect \rho\in \R_+^{\N_0}$ and $\sigma \geq 0$ with $\sum_{j=0}^\infty \rho_j \leq \sigma \leq 1$. Then 
the limit~\eqref{eq:entropydef} exists and is given by 
	$$
			s\bigl( (\rho_j)_{j\in \N_0},\sigma\bigr)
			=  - \sum_{j=0}^\infty \frac{1}{|B_j|} \Bigl( \rho_j \log \frac{\rho_j}{1- \sigma_{j+1}} + (1- \sigma_j) \log \frac{1- \sigma_j}{1 - \sigma_{j+1}} \Bigr)
	$$	
	with the convention $0 \log \frac00 =0$. Moreover 
		$$
			0 \leq s(\vect \rho,\sigma) \leq \sum_{j=0}^\infty \frac{1-\sigma_{j+1}}{|B_j|}\log 2 < \infty.
		$$
\end{theorem}

\noindent An equivalent expression in terms of effective activities $\widehat \rho_j$ is given in Eq.~\eqref{eq:hatentropy} below. Notice that the entropy vanishes if $\rho_j = 0$ for all $j\in \N_0$---only small blocks (i.e., blocks whose size does not scale with the volume) contribute to the entropy. 

\begin{proof}
	Configurations can be constructed by placing first the biggest block (if present), i.e., $n$-blocks, then blocks of type $n-1$, etc. 
The entropy equals
 	 $$
	 	S_{\Lambda_n} (N_0^{(n)},\ldots, N_n^{(n)}) = \sum_{j=0}^n \log \binom{ (|\Lambda_n| - \sum_{k=j+1}^n |B_k|\, N_k^{(n)})/|B_j|}{N_j^{(n)}}.
	 $$
	 Indeed, having chosen the blocks of $\omega$ of type $k\geq j+1$, there are $(|\Lambda_n| - N_n^{(n)} |B_n| - \cdots - N_{j+1}^{(n)}|B_j| )/|B_j|$ available $j$-blocks to choose from for the placement of the next $N_j^{(n)}$  blocks of type $j$.  	
	 
	 Set $\rho_j^{(n)} := N_j^{(n)}|B_j| / |\Lambda_n|$ and $\sigma_j^{(n)}:= \sum_{k=j}^n \rho_k^{(n)}$. 
	Clearly $\rho_j^{(n)}\to \rho_j$ and $\sigma_j^{(n)}\to \sigma$ for  all $j\in \N_0$.  Stirling's formula and the resulting approximation $\log \binom{m}{k} = - k \log\frac{k}{m} - (m - k) \log (1- \frac k m ) + O(\log k ) + O(\log(m-k)) + O(\log m)$ yield
	$$
	\frac{1}{|\Lambda_n|}	S_{\Lambda_n} (N_1^{(n)},\ldots, N_n^{(n)})   = - \sum_{j=0}^n \frac{1}{|B_j|} \Bigl( \rho_j^{(n)} \log \frac{\rho_j^{(n)}}{1- \sigma_{j+1}^{(n)}} + (1- \sigma_j^{(n)}) \log \frac{1- \sigma_j^{(n)}}{1 - \sigma_{j+1}^{(n)}} \Bigr) + o(1).
	$$
	 Summation and limits can be exchanged because each summand is bounded in absolute value by $\frac{1- \sigma_{j+1}}{|B_j|} (\log 2)$ (see Eq.~\eqref{eq:hatentropy} below) and $\sum_{j} \frac{1}{|B_j|}<\infty$. The proposition follows. 
\end{proof} 

\noindent The proof of Theorem~\ref{thm:entropy} suggests to work with effective densities. Set
$$
	 \widehat \rho_j := \frac{\rho_j}{1- \sigma_{j+1}} = \frac{\rho_j}{ 1- \sum_{k=j+1}^\infty \rho_k - \sigma_\infty}
$$
with $\sigma_j$ and $\sigma_\infty$ defined in~\eqref{eq:sigmajdef}. 
Thus $\widehat \rho_j$ takes into account the volume excluded by cubes of type $k\geq j+1$. The entropy becomes 
\be \label{eq:hatentropy} 
	s\bigl( (\rho_j)_{j\in \N_0},\sigma\bigr)
	= - \sum_{j=0}^\infty \frac{1- \sigma_{j+1}}{|B_j|} \Bigl( \widehat \rho_j \log \widehat \rho_j +(1-\widehat \rho_j )\log (1-\widehat \rho_j)\Bigr).
\ee
The entropy for the ideal mixture, where cubes may overlap, is instead given by 
\be \label{eq:idealentropy}
	s^\mathrm{Ber}\bigl( (\rho_j)_{j\in \N_0},\sigma\bigr)
	= - \sum_{j=0}^\infty \frac{1}{|B_j|} \Bigl(  \rho_j \log  \rho_j +(1-\rho_j )\log (1- \rho_j)\Bigr).
\ee
The expressions for the entropy are again very similar to each other, just as for the pressure. The similarity in equations can be pushed a bit further. In the multi-canonical ensemble we define the chemical potential of $j$-blocks by 
\be \label{eq:chem1}
	\mu_j \bigl( (\rho_j)_{j\in \N_0},\sigma_\infty \bigr):=- |B_j| \frac{\partial }{\partial \rho_j}s\Bigl( (\rho_j)_{j\in \N_0},\sigma_\infty + \sum_{j=0}^\infty \rho_j \Bigr).
\ee
The chemical potential can be thought of as a derivative with respect to $\nu_j = \rho_j / |B_j|$, which is the expected number of $j$-blocks per unit volume (remember~\eqref{eq:nujrhoj}). The derivative is taken at constant $\sigma_\infty$ rather than constant $\sigma$. We also define
\be \label{eq:chem2}
	\mu_\infty \bigl( (\rho_j)_{j\in \N_0},\sigma_\infty \bigr):=- \frac{\partial }{\partial \sigma_\infty} s\Bigl( (\rho_j)_{j\in \N_0},\sigma_\infty + \sum_{j=0}^\infty \rho_j \Bigr).
\ee
Explicit computations yield 
\be \label{eq:chem3}
	\mu_j = \log \frac{\widehat \rho_j}{1- \widehat \rho_j} - |B_j| \sum_{k=0}^{j-1}\frac{1}{|B_k|} \log(1- \widehat \rho_k) , \qquad \mu_\infty = - \sum_{j=0}^\infty \frac{1}{|B_j|} \log (1- \widehat \rho_j). 
\ee
For the Bernoulli mixture, in contrast, 
$$
	\mu_j^\mathrm{Ber} = \log \frac{ \rho_j}{1- \rho_j},\qquad \mu_\infty ^\mathrm{Ber} = 0. 
$$
The chemical potentials coincide up to error terms of order $O(\sum_j \rho_j) + O(\sigma_\infty) = O(\sigma)$.

\subsection{Analyticity. Multi-species virial expansion} 

Before we turn to a variational representation of the pressure, we collect a few analytic properties of the entropy that are of intrinsic interest. Consider the complex Banach space $\ell^1(\N_0)\times \C$ with norm $||(\vect \rho,\sigma_\infty)|| = \sum_{j=0}^\infty |\rho_j| + |\sigma_\infty|$ and the open unit ball 
$B(0,1)= \{ (\vect \rho,\sigma_\infty):\ ||(\vect \rho,\sigma_\infty)|| <1 \}$. 
Define $\sigma_j=\sigma_\infty+ \sum_{k=j}^\infty \rho_k$ and 
\be \label{eq:taylor}
	\Phi\bigl( \vect \rho,\sigma_\infty) :=\sum_{m=2}^\infty \frac{1}{m(m-1)} \sum_{j=0}^\infty \frac{1}{|B_j|} \bigl( \sigma_j^{m} - \sigma_{j+1}^m\bigr).
\ee

\begin{prop} \label{prop:analyticity} \hfill
	\begin{enumerate} 
		\item [(a)] 	The map $\Phi$ is holomorphic in the open unit ball and the  Taylor series~\eqref{eq:taylor} converges uniformly in every open ball $B(0,r)$ of radius $r<1$. 
		\item [(b)] The entropy satisfies 
		$$
			s(\vect \rho,\sigma_\infty) = - \sum_{j=0}^\infty \frac{1}{|B_j|} \rho_j(\log \rho_j - 1) - \Phi(\vect \rho,\sigma_\infty)
		$$
		for all $(\vect\rho,\sigma_\infty) \in \R_+^{\N_0}\times \R_+$ with $\sum_{j=0}^\infty \rho_j + \sigma_\infty < 1$. 
	\end{enumerate} 
\end{prop} 

\noindent A short overview and list of references on holomorphic functions in Banach spaces is provided in~\cite[Appendix B]{jansen-kuna-tsagkaro2019}.

\begin{proof}
 We compute, using $\sigma_j = \rho_j + \sigma_{j+1}$, 
	\begin{align*}
		& \rho_j\log \frac{\rho_j}{1-\sigma_{j+1}} + (1- \sigma_j)\log \frac{1- \sigma_j}{1-\sigma_{j+1}} \\
		&\qquad  = \rho_j \log\rho_j + (1- \sigma_j  ) \log (1- \sigma_j) - (1- \sigma_{j+1}) \log (1- \sigma_{j+1})\\
		&\qquad =  \rho_j\bigl( \log\rho_j-1\bigr) + (1- \sigma_j  ) \Bigl(\log (1- \sigma_j) -1\Bigr)- (1- \sigma_{j+1}) \Bigl(\log (1- \sigma_{j+1})-1\Bigr). 
	\end{align*} 
	Because of 
	$$
		(1- x)\Bigl( \log (1-x) - 1\Bigr) = -1 - \int_0^x \log(1- y) \dd y = - 1+  \sum_{m=2}^\infty \frac{x^m}{m(m-1)} \qquad (|x|<1),
	$$
	we deduce that the $j$-th summand in the formula for the entropy from Theorem~\ref{thm:entropy} is given by 
	\be \label{eq:ente}
		- \frac{1}{|B_j|} \rho_j(\log \rho_j-1) -  \frac{1}{|B_j|} \sum_{m=2}^\infty \frac{1}{m(m-1)} (\sigma_j^{m}- \sigma_{j+1}^m). 
	\ee
	In order to split the series over $j$ into two contributions corresponding to the two terms in the preceding sum, we need to check that the two sums are absolutely convergent. For the first term, we note that $\sup_{x\in [0,1]} |x(\log x - 1)|=1$ hence 
	$$
		\sum_{j=0}^\infty \frac{1}{|B_j|} \bigl|  \rho_j (\log \rho_j -1 )\bigr| \leq \sum_{j=0}^\infty \frac{1}{|B_j|}<\infty. 
	$$
	For the convergence of $\Phi$, corresponding to the second term in~\eqref{eq:ente} set 
	$$
		P_m(\vect \rho,\sigma_\infty):= \frac{1}{m(m-1)} \sum_{j=0}^\infty \frac{1}{|B_j|} \bigl( \sigma_j^{m} - \sigma_{j+1}^m\bigr).
	$$	
	Because of 
	$$
		\bigl|\sigma_j^{m} - \sigma_{j+1}^m \bigr|= \Biggl| \rho_j \sum_{k=0}^{m-1} \sigma_j^k \sigma_{j+1}^{m-1-k} \Biggr| \leq m |\rho_j|\, ||(\vect \rho,\sigma)||^{m-1}
	$$
	and $|B_j|\geq 1$, 	we have 	
	$$
		\bigl|P_m(\vect \rho,\sigma_\infty)\bigr| \leq \frac{1}{m-1}\Bigl( \sum_{j=0}^\infty \frac{1}{|B_j|} |\rho_j|\Bigr) ||(\vect \rho,\sigma_\infty)||^{m-1}\leq ||(\vect \rho,\sigma_\infty)||^{m}< \infty. 
	$$
	It follows that  $P_m$ is absolutely convergent in $B(0,1)$ and defines a continuous $m$-homogeneous polynomial with norm 
	$$
		||P_m|| =\sup_{||(\vect \rho,\sigma_\infty)||\leq 1} |P_m(\vect \rho,\sigma_\infty)| \leq 1,
	$$
	moreover $\Phi(\vect \rho,\sigma_\infty) = \sum_{m=2}^\infty P_m(\vect \rho,\sigma_\infty)$ converges uniformly in $||(\vect \rho,\sigma_\infty)||\leq r$, for every $r\in (0,1)$. This proves the analyticity in the open unit ball. The formula for the entropy follows  from~\eqref{eq:ente}. 
\end{proof} 	

\subsection{Variational representation for the pressure}

\begin{prop} \label{prop:pressure-legendre}
	Assume that $\lim_{j\to \infty} \frac{1}{|B_j|} \log z_j=\theta^*$. Then the pressure has the variational representation
	$$
	p \bigl( (z_j)_{j\in \N_0}\bigr) = \sup\Biggl\{ \sum_{j=0}^\infty \frac{\rho_j}{|B_j|}\log z_j +\Bigl( \sigma - \sum_{j=0}^\infty \rho_j\Bigr) \theta^* + s\bigl( (\rho_j)_{j\in \N_0}, \sigma)  \, \Bigr|\, \sum_{j=0}^\infty \rho_j \leq \sigma \leq 1\Biggr\}.
	$$
	In addition:
	\begin{enumerate}
		\item [(a)] If $\sum_{j=0}^\infty \widehat z_j<\infty$ and $p((z_j)_{j\in \N_0})>\theta^*$, then the tuple $(\vect \rho(\vect z) ,\sigma(\vect z))$ given in Theorem~\ref{thm:density}(a) is the unique maximizer. It satisfies $\sigma_\infty =0$ and $\sigma<1$. 
		\item [(b)] If $\sum_{j=0}^\infty \widehat z_j<\infty$ and $p((z_j)_{j\in \N_0})= \theta^*$, then the set of maximizers is given by the convex combinations of $(\vect \rho(z),\sigma(\vect z))$ from Theorem~\ref{thm:density}(a) and  $(\vect 0, 1)$.
		\item [(c)] If $\sum_{j=0}^\infty \widehat z_j=\infty$, then $p((z_j)_{j\in \N_0})= \theta^*$ and the unique maximizer is the tuple $(\vect 0, 1)$.
	\end{enumerate}
\end{prop} 

\noindent We leave as an open problem whether the proposition extends to  activities with $\liminf_{j\to \infty} \frac{1}{|B_j|} \log z_j<\limsup_{j\to \infty} \frac{1}{|B_j|} \log z_j=\theta^*$. The cases (a), (b), and (c) correspond to a gas phase, coexistence region, and condensed phase, respectively.

\begin{proof} [Proof of the variational formula in Proposition~\ref{prop:pressure-legendre}]
	Let $(\rho_j)_{j\in \N_0}\in\R_+^{\N_0}$ and $\sigma\in [0,1]$ with $\sum_{j=0}^\infty \rho_j \leq \sigma$. Then there exist sequences $N_j^{(n)}$ of integers satisfying~\eqref{eq:Nconv}. Clearly 
	\be \label{eq:pe1}
		\log \Xi_{\Lambda_n} \geq \sum_{j=0}^n N_j^{(n)} \log z_j + S_{\Lambda_n}\bigl(N_1^{(n)},\ldots,N_n^{(n)}\bigr).
	\ee
	The second term, divided by $|\Lambda_n|$, converges to $s((\rho_j)_{j\in \N_0},\sigma)$ by Theorem~\ref{thm:entropy}. For the first term, we set $z'_j:=z_j \exp( - |B_j|\theta^*)$ and we write for $n\geq k$ 
	\begin{align*}
		& \Biggl|\sum_{j=0}^n \frac{N_j^{(n)}}{|\Lambda_n|} \log z_j - \sum_{j=0}^\infty \frac{\rho_j}{|B_j|} \log z_j - \Bigl( \sigma - \sum_{j=0}^\infty \rho_j\Bigr) \theta^*\Biggr| \\
		& \qquad \leq \Biggl|\sum_{j=0}^n \frac{N_j^{(n)}}{|\Lambda_n|} \log z'_j - \sum_{j=0}^\infty \frac{\rho_j}{|B_j|} \log z'_j - \Bigl( \sigma - \sum_{j=0}^n \frac{N_j^{(n)}|B_j|}{|\Lambda_n|} \Bigr) \theta^*\Biggr| \\
		& \qquad \leq \sum_{j=0}^k \Bigl| \frac{N_j^{(n)}}{|\Lambda_n|}  - \frac{\rho_j}{|B_j|} \Bigr| |\log z'_j| + 2 \max_{j\geq k+1}\Bigl| \frac{1}{|B_j|}\log z'_j\Bigr| + |\theta^*|\,\Bigl|\sigma - \sum_{j=0}^n \frac{N_j^{(n)}|B_j|}{|\Lambda_n|}\Bigr|.
	\end{align*}
	Taking first the limit $n\to \infty$ and then $k\to \infty$, we see that overall the expression goes to zero. Turning back to~\eqref{eq:pe1} we get 
	$$
		\liminf_{n\to \infty} p_{\Lambda_n}\geq \sum_{j=0}^\infty \frac{\rho_j}{|B_j|}\log z_j +\Bigl( \sigma - \sum_{j=0}^\infty \rho_j\Bigr) \theta^* + s\bigl( (\rho_j)_{j\in \N_0}, \sigma).
	$$
	This holds true for all $(\rho_j)_{j\in \N_0}$ and $\sigma\in [0,1]$ with $\sum_{j=0}^\infty \rho_j \leq \sigma$, accordingly the limit inferior of the pressure is bounded from below by a supremum. 	
	
	For the upper bound, let $\mathcal I_n \subset \N_0^n$ be the set of vectors $(N_1^{(n)},\ldots,N_n^{(n)})$ with
	$\sum_{j=0}^n |B_j|\, N_j^{(n)}\leq |\Lambda_n|$. Every such vector is uniquely identified with an integer partition of $|\Lambda_n|$, therefore by the Hardy-Ramanujan formula 
	\be \label{eq:hardy-rama}
		|\mathcal I_n| \leq \exp\Bigl( o\bigl(|\Lambda_n|\bigr)\Bigr).
	\ee
	Clearly 
	\be \label{eq:xiupper}
		\Xi_{\Lambda_n}\leq |\mathcal I_n|\, \max_{(N_1^{(n)},\ldots,N_n^{(n)})\in \mathcal I_n}\exp \Biggl(\sum_{j=0}^n N_j^{(n)} \log z_j + S_{\Lambda_n}\bigl(N_1^{(n)},\ldots,N_n^{(n)}\bigr)\Biggr).
	\ee
	Consider the sequence of maximizers of the right-hand side.  By compactness, every subsequence admits in turn a subsequence that satisfies~\eqref{eq:Nconv} for some $(\rho_j)_{j\in\N_0}$ and $\sigma \in [0,1]$ with $\sum_{j=0}^\infty \rho_j\leq \sigma$. The proof of the upper bound for the limit superior of the pressure is easily completed by combining Eqs.~\eqref{eq:hardy-rama}, \eqref{eq:xiupper}, and arguments similar to the proof of the lower bound. This proves the variational representation of the pressure.
\end{proof}

\noindent The proof of items (a) and (b) in Proposition~\ref{prop:pressure-legendre} builds on several lemmas. First we show that for $\sigma_\infty =0$, the expression to be maximized is a combination of relative entropies of measures on $\{0,1\}$, corresponding to absence or presence of a cube. 

\begin{lemma} \label{lem:shannon}
	 For every $(\rho_j)_{j\in \N_0}\in \R_+^{\N_0}$ and $
	\sigma \in [0,1]$ with $\sum_{j=0}^\infty \rho_j=\sigma$ (equivalently, $\sigma_\infty =0$), we have 
\begin{multline} \label{eq:relative-entropy}
	p\bigl((z_j)_{j\in \N_0}\bigr) - \Biggl(\sum_{j=0}^\infty \frac{\rho_j}{|B_j|}\log z_j + s\bigl( (\rho_j)_{j\in \N_0}, \sigma) \Biggr)\\
	= - \sum_{j=0}^\infty \frac{1-\sigma_{j+1}}{|B_j|} \Biggl( \widehat \rho_j \log \frac{\widehat \rho_j}{\widehat z_j/(1+\widehat z_j)}+ (1-\widehat \rho_j )\log \frac{1-\widehat \rho_j}{1/(1+\widehat z_j)}\Biggr).
\end{multline}
\end{lemma}

\begin{proof}
We  compute
\begin{align*}
	\sum_{j=0}^\infty \frac{\rho_j}{|B_j|}\log z_j &= 	\sum_{j=0}^\infty \frac{\rho_j}{|B_j|}\Bigl( \log \widehat z_j + |B_j|\sum_{k=0}^{j-1} \frac{1}{|B_k|}\log(1+\widehat z_k)\Bigr)\\
	&= \sum_{j=0}^\infty \frac{\rho_j}{|B_j|} \log \widehat z_j +\sum_{k=0}^\infty \frac{1}{|B_k|}\log(1+\widehat z_k) \sum_{j=k+1}^\infty \rho_j \\
	&=\sum_{j=0}^\infty  \frac{1-\sigma_{j+1}}{|B_j|} \widehat \rho_j \log \widehat z_j+ \sum_{k=0}^\infty \frac{\sigma_{k+1}}{|B_k|}\log(1+\widehat z_k).
\end{align*}
In going from the second to the third line we have used the equality $\sum_{j=k+1}^\infty \rho_k = \sigma_{k+1}$, which is valid because of $\sigma_\infty =0$. It follows that 
$$
	p\bigl((z_j)_{j\in \N_0}\bigr) - \sum_{j=0}^\infty \frac{\rho_j}{|B_j|}\log z_j = \sum_{j=0}^\infty \frac{1-\sigma_{j+1}}{|B_j|} \Bigl( \log(1+\widehat z_j) -\widehat \rho_j \log \widehat z_j\Bigr).
$$
We combine with the formula for the entropy from Theorem~\ref{thm:entropy} and obtain~\eqref{eq:relative-entropy}.
\end{proof}

\noindent The term in parentheses on the right-hand side of~\eqref{eq:relative-entropy}, together with the minus sign, is nothing else but the relative entropy of the Bernoulli measure with parameter $\widehat \rho_j$ with respect to the Bernoulli measure with parameter $\widehat z_j/(1+\widehat z_j)$. It is non-negative and vanishes if and only if $\widehat \rho_j =\widehat z_j/(1+\widehat z_j)$. The  next lemma relates this identity to Theorem~\ref{thm:density}.

\begin{lemma}\label{lem:equiv}
	Let $(\rho_j)_{j\in \N_0}\in \R_+^{\N_0}$ and $\sigma:= \sum_{j=0}^\infty \rho_j$. Pick $m\in \N_0$ and assume $\sigma_{m+1}=\sum_{j=m+1}^\infty \rho_j< 1$. Then the following two statements are equivalent:
	\begin{itemize}
		\item[(i)] $\widehat \rho_j =\widehat z_j/(1+\widehat z_j)$ for all $j\geq m$.
		\item [(ii)] $\rho_j = \widehat z_j \prod_{k=j}^\infty (1+\widehat z_k)^{-1}$ for all $j\geq m$.
	\end{itemize}
\end{lemma}

\noindent Let us stress that the lemma works both for $\sum_j \widehat z_j < \infty$ and $\sum_j \widehat z_j = \infty$. In the latter case the infinite products vanish and we find $\rho_j =0$ for all $j\geq m$.

\begin{proof}
	We note
	$$
		1-\sigma_j= 1- \sigma_{j+1}-\rho_j = (1- \sigma_{j+1})(1- \widehat \rho_j)
	$$		
	hence 
	$1- \sigma_j=(1- \sigma_\ell) \prod_{k=j}^{\ell-1}(1-\widehat\rho_j)$ for all $\ell \geq j \geq m$. Because of $\sum_{j=0}^\infty\rho_j =\sigma$ we have $\sigma_\infty=0$ and $\lim_{\ell\to \infty}\sigma_\ell =0$, hence
	$$
		1- \sigma_j = \prod_{k=j}^\infty(1-\widehat\rho_j).
	$$	
	If (i) holds true, then for all $j\geq m$ 
	$$
		\rho_j = (1-\sigma_{j+1}) - (1-\sigma_j) = \widehat\rho_j \prod_{k=j+1}^\infty (1- \widehat \rho_k).
	$$
	The implication (i) $\Rightarrow$ (ii) follows. Conversely, if (ii) holds, let $Y_j$ be independent Bernoulli variables with $\P(Y_j=0) =1/(1+\widehat z_j)$. Then
	$$
		\rho_j=\P(Y_j=1,\,\forall k\geq j+1: Y_k=0)
	$$
	and
	$$
		1 - \sigma_r = 1-\P(\exists j\geq r:\, Y_j=1)=\P(\forall j \geq r:\, Y_j = 0)=\prod_{j=r}^\infty \frac 1{1+\widehat z_j}
	$$
	and (i) follows.
\end{proof}

\noindent The previous two lemmas deal with the gas phase ($\sigma_\infty=0$) only. The next lemma allows for $\sigma_\infty \geq 0$ and is particularly relevant for the coexistence region.  Let us briefly motivate a new set of variables. Suppose that $\sigma_\infty\in (0,1)$. Then we may think of the system as a mixture of a condensed phase, occupying the volume fraction $\sigma_\infty$, and a gas phase in the remaining volume fraction $1-\sigma_\infty$. The natural density variables for the gas phase should be defined relatively to the volume occupied by the gas and not the total volume. Therefore we introduce the new variables 
\be\label{eq:primed}
		\rho'_j:= \frac{\rho_j}{1- \sigma_\infty}, \quad \sigma':=\sum_{j=0}^\infty \rho'_j,\quad \sigma'_j :=  \sum_{k=j}^\infty \rho'_j.
\ee

\begin{lemma} \label{lem:convex-combi}
	Let $((\rho_j)_{j\in \N_0}, \sigma)\in \R_+^{\N_0}\times [0,1]$ with $\sum_{j=0}^\infty \rho_j \leq \sigma$ and $\sigma_\infty\in (0,1)$. 
	Then 
	\begin{multline*}
		\sum_{j=0}^\infty \frac{\rho_j}{|B_j|}\log z_j +\Bigl( \sigma - \sum_{j=0}^\infty \rho_j\Bigr) \theta^* + s\bigl( (\rho_j)_{j\in \N_0}, \sigma) \\
			= (1- \sigma_\infty) \Biggl( 		\sum_{j=0}^\infty \frac{\rho'_j}{|B_j|}\log z_j + s\bigl( (\rho'_j)_{j\in \N_0}, \sigma')\Biggr) +\sigma_\infty \theta^*.
	\end{multline*} 
\end{lemma} 

\noindent Put differently, the grand potential in the coexistence region is a convex combination of the grand potential $\theta^*$ in the condensed phase and the grand potential of the gas phase. 

\begin{proof}
	The lemma follows from Theorem~\ref{thm:entropy} and explicit computations. Clearly 
	$$
		\sum_{j=0}^\infty \frac{\rho_j}{|B_j|}\log z_j +\Bigl( \sigma - \sum_{j=0}^\infty \rho_j\Bigr) \theta^* = (1- \sigma_\infty) 		\sum_{j=0}^\infty \frac{\rho'_j}{|B_j|}\log z_j + \sigma_\infty \theta^*,
	$$
	so it remains to check that 
	\be \label{eq:primedent}
		s\bigl( (\rho_j)_{j\in \N_0},\sigma\bigr) = (1- \sigma_\infty) 		s\bigl( (\rho'_j)_{j\in \N_0},\sigma'\bigr).
	\ee
	As a preliminary observation we note $\sigma' = (\sigma-\sigma_\infty) / (1- \sigma_\infty) \leq1$.  In view of 
	$$
		1 - \sigma_{j+1} = 1 - \sum_{k=j+1}^\infty \rho_j - \sigma_\infty = (1- \sigma_\infty)(1- \sigma'_{j+1}), 
	$$	
	we also have $\rho'_j \leq 1 - \sigma'_{j+1}$, moreover 
	\begin{align*} 
		s\bigl( (\rho_j)_{j\in \N_0},\sigma\bigr) & =  - (1- \sigma_\infty)  \sum_{j=0}^\infty \frac{1}{|B_j|} \Bigl( \rho'_j \log \frac{\rho'_j}{1- \sigma'_{j+1}} + (1- \sigma'_j) \log \frac{1- \sigma'_j}{1 - \sigma'_{j+1}} \Bigr) \\
		& = (1- \sigma_\infty) 	s\bigl( (\rho'_j)_{j\in \N_0},\sigma'\bigr). \qedhere
	\end{align*} 
\end{proof} 

\begin{proof}[Proof of Proposition~\ref{prop:pressure-legendre}(a)-(c)]
Assume $\sum_{j=0}^\infty\widehat z_j <\infty$ and $p((z_j)_{j\in \N_0})>\theta^*$. To prove part (a), we proceed in two steps: First we show that a tuple $((\rho_j)_{j\in \N_0}, \sigma)$ with $\sigma_\infty =0$, i.e., $\sum_{j=0}^\infty \rho_j = \sigma$, is a maximizer if and only if it is given by the expressions from~Theorem~\ref{thm:density}(a). Second, we show that every maximizer necessarily satisfies $\sigma_\infty =0$. 

For Step 1, we use Lemma~\ref{lem:shannon}. A tuple with $\sigma_\infty =0$ is a maximizer if and only if the right-hand side of~\eqref{eq:relative-entropy} vanishes. But on the right-hand side of~\eqref{eq:relative-entropy}, the term in parentheses, together with the minus sign, is nothing else but the relative entropy of two Bernoulli measures with parameters $\widehat \rho_j$ and  $\widehat z_j/(1+\widehat z_j)$. As a consequence the overall sum vanishes---i.e., the tuple $(\rho_j)_{j\in \N_0}$, $\sigma=\sum_{j=0}^\infty \rho_j$ is a maximizer---if and only if, for every $j\in \N_0$, we have $\sigma_{j+1}=1$ or $\widehat \rho_j =\widehat z_j/(1+ \widehat z_j)$. 

Suppose by contradiction that there is a maximizer with $\sigma_{r+1}=1$ for some $r\in \N_0$, and $\sigma=\sum_{j=0}^\infty  \rho_j$. The sequence $(\sigma_j)$ is monotone decreasing, therefore if the set of such  $r$'s is unbounded, then $\sigma_j=1$ for all $j\in \N_0$. It follows that $\rho_j =\sigma_j - \sigma_{j+1}=0$ for all $j$ and $\sigma_{j+1}=\sum_{k=r+1}^\infty \rho_j=0$, contradiction.
Thus the set of $r$'s with $\sigma_{r+1}=1$ is bounded, let $m$ be its maximal element. Then $\sigma_{m+1}=\sum_{k=m+1}^\infty \rho_k=1$ hence $\rho_0=\cdots=\rho_{m}=0$. In addition, $\sigma_{j+1}<1$ and $\widehat \rho_j =\widehat z_j/(1+ \widehat z_j)$ for all $j\geq m+1$. It follows that for all $j\geq m+1$, the density $\rho_j$ is given by the formula from Theorem~\ref{thm:density}(a), see Lemma~\ref{lem:equiv}.  In particular, 
$\sigma_{m+1}=\sum_{j=m+1}^\infty \rho_j$ is bounded by the packing fraction from Theorem~\ref{thm:density}(a), which is strictly smaller than $1$. Thus $\sigma<1$, in contradiction with $\sigma= \sigma_{m+1}=1$.

Consequently $\sigma_{j+1}<1$ and $\widehat \rho_j =\widehat z_j/(1+ \widehat z_j)$ for all $j\in \N_0$. Lemma~\ref{lem:equiv} shows that the maximizer is given by the formulas from Theorem~\ref{thm:density}(a). In particular, $\sigma<1$ and $\sigma_\infty =0$. 

For Step 2, we use Lemma~\ref{lem:convex-combi}. Let $((\rho_j)_{j\in \N_0},\sigma)$ be such that $\sigma_\infty >0$. By Lemma~\ref{lem:convex-combi} and the preceding considerations applied to $((\rho'_j)_{j\in \N_0},\sigma')$, we can bound 
\be \label{ineq:convex}
	\sum_{j=0}^\infty \frac{\rho_j}{|B_j|}\log z_j +\Bigl( \sigma - \sum_{j=0}^\infty \rho_j\Bigr) \theta^* + s\bigl( (\rho_j)_{j\in \N_0}, \sigma) 
	\leq (1- \sigma_\infty) p\bigl( (z_j)_{j\in \N_0}\bigr) + \sigma_\infty \theta^*
\ee
which is strictly smaller than $p\bigl( (z_j)_{j\in \N_0}\bigr)$ because of the assumption $\theta^*< p\bigl( (z_j)_{j\in \N_0}\bigr)$. Therefore the tuple is not a maximizer. This concludes Step 2 and the proof of part (a) of the proposition. \\

\medskip \noindent 
For (b) and (c), assume $p((z_j)_{j\in \N_0})=\theta^*$. Then $(\vect \rho, \sigma) = (\vect 0, 1)$ is a maximizer. Suppose that there exists another maximizer $(\vect \rho, \sigma)$. Then necessarily $\sigma_\infty < 1$ and we may define primed variables $(\vect {\rho'},\sigma')$ and $\sigma'_j$ as in Eq.~\eqref{eq:primed}. The variational representation for the pressure, the equality $p((z_j)_{j\in \N_0})=\theta^*$, and Lemma~\ref{lem:convex-combi} yields 
\begin{multline*}
	0 = \theta^* - \Bigl( \sum_{j=0}^\infty \frac{\rho_j}{|B_j|}\log z_j +\Bigl( \sigma - \sum_{j=0}^\infty \rho_j\Bigr) \theta^* + s\bigl( (\rho_j)_{j\in \N_0}, \sigma)\Bigr)  \\
	= (1- \sigma_\infty) \Biggl\{ \theta^* - \Bigl(\sum_{j=0}^\infty \frac{\rho'_j}{|B_j|}\log z_j + s\bigl( (\rho'_j)_{j\in \N_0}, \sigma')\Bigr) \Biggr\}\geq 0
\end{multline*} 
hence 
\be \label{eq:bctheta}
	\theta^* - \Bigl(\sum_{j=0}^\infty \frac{\rho'_j}{|B_j|}\log z_j + s\bigl( (\rho'_j)_{j\in \N_0}, \sigma')\Bigr) =0. 
\ee
Since $p((z_j)_{j\in \N_0})=\theta^*$, the left-hand side can be expressed as a combination of relative entropies of Bernoulli variables as in Lemma~\ref{lem:shannon}. 

Assume first $\sum_{j=0}^\infty \widehat z_j <\infty$. Adapting the arguments of the proof of part (a) we deduce 
$$
	\rho'_j = \rho_j(\vect z)= \frac{\widehat z_j}{1+\widehat z_j}\prod_{k=j+1}^\infty \frac{1}{1+\widehat z_k}\quad (j\in \N_0). 
$$
Then $\rho_j = (1- \sigma_\infty) \rho'_j$ and 
$$
	\sigma = \sum_{j=0}^\infty \rho_j + \sigma_\infty = (1- \sigma_\infty) \sigma' + \sigma_\infty 
$$
by definition of $\rho'_j$ and $\sigma'$. It follows that the additional maximizer $(\vect \rho,\sigma)$ is a convex combination of $(\vect \rho(\vect z),\sigma(\vect z))$ and $(\vect 0, 1)$. Conversely, every such convex combination is indeed a maximizer. This proves part (b) of Proposition~\ref{prop:pressure-legendre}. 

If on the other hand $\sum_{j=0}^\infty\widehat z_j =\infty$, then we check that $\rho'_j = 0$ hence $\rho_j =0$ for all $j$. To that aim we revisit the arguments from the proof of part (a). We start from~\eqref{eq:bctheta} and 
deduce as in part (a) that  $\sigma'_{j+1}=1$ or $\widehat \rho'_j =\widehat z_j/(1+ \widehat z_j)$ for all $j\in \N_0$. We distinguish several cases. 

If $\sigma'_{j+1}=1$ for all $j\in \N_0$, then $\rho'_j =0$ for all $j\in \N_0$ and $\sigma'=0$, contradicting $\sigma'_{j+1}=1$. 

If $\sigma'_{j+1}\neq 1$ for some $j$, then the set $\{r\in \N_0 \mid \sigma'_{r+1}=1\}$ is bounded. Suppose by contradiction that it is non-empty and let $m$ be its maximum. 
Then $\sigma'_{m+1}=\sum_{k=m+1}^\infty \rho'_k=1$ hence $\rho'_0=\cdots=\rho'_{m}=0$. In addition, $\sigma'_{j+1}<1$ and $\widehat {\rho'}_j =\widehat z_j/(1+ \widehat z_j)$ for all $j\geq m+1$. 
Lemma~\ref{lem:equiv} yields $\rho'_j =0$ for all $j \geq m+1$. It follows that $\sigma'_{m+1}=0$, in contradiction with the identity  $\sigma'_{m+1}=1$ that holds true by definition of $m$. 

The only case left is $\sigma'_{j+1}<1$ for all $j \in \N_0$. In this case Lemma~\ref{lem:equiv} again yields $\rho'_j =0$ for all $j \in \N_0$ hence $\sigma'=0$. 

Consequently $\rho_j = (1-\sigma_\infty) \rho_j =0$ for all $j \in \N_0$ and $\sigma = \sigma_\infty$. The grand-potential of such a configuration is $\sigma_\infty \theta^*$, which is equal to $\theta^*$ if and only if $\sigma_\infty =1$. As a consequence, $(\vect 0,1)$ is the unique maximizer of the grand potential. This proves part (c). 
\end{proof}

\section{Phase transition} \label{sec:phasetransition}

\subsection{Generalities. Parameter-dependent activity} 

Let $(E_j)_{j\in \N_0}$ be a sequence in $\R\cup \{\infty\}$ such that $E_j/|B_j|$ has a limit in $\R\cup \{\infty\}$, i.e., 
$$
	e_\infty:= \lim_{j\to \infty} \frac{E_j}{|B_j|} > - \infty, 
$$
and $E_j <\infty$ for at least one $j\in \N_0$. Think of $E_j$ as the energy of a block, which could be a bulk contribution plus a boundary term, e.g.,  $E_j = e_\infty |B_j| + \mathrm{const} |\partial B_j|$. For later purpose we also define 
$$
	E(B) = E_j \quad (B\in \mathbb B_j).
$$
We specialize to parameter-dependent activities of the form 
$$
	z_j(\mu) = \exp\bigl(|B_j| \mu - E_j\bigr) \qquad (\mu \in \R). 
$$
The activity is stable with 
\be \label{eq:thetamu} 
	\theta^*(\mu) = \lim_{j\to \infty}\frac{1}{|B_j|} \log z_j(\mu) = \mu - e_\infty. 
\ee
We write $p(\mu)$, $\widehat z_j(\mu)$, $\rho_j(\mu)$ for the pressure, effective activities, and density variables of the $\mu$-dependent model. For $(\rho_j)_{j\in \N_0}\in \R_+^{\N_0}$ and $\sigma_\infty \geq 0$ with $\sum_{j=0}^\infty \rho_j + \sigma_\infty \leq 1$, define the free energy of a block size distribution 
\be \label{eq:freenergy} 
	f\bigl( (\rho_j)_{j\in \N_0}, \sigma_\infty \bigr) 
	:= \sum_{j=0}^\infty \frac{\rho_j}{|B_j|}\, E_j + \sigma_\infty e_\infty -  s\Bigl( (\rho_j)_{j\in \N_0}, \sigma_\infty + \sum_{j=0}^\infty \rho_j \Bigr).  
\ee
and the free energy at given packing fraction $\sigma\in [0,1]$
$$
	\varphi(\sigma) = \inf \Bigl\{ 	f\bigl( (\rho_j)_{j\in \N_0}, \sigma_\infty \bigr)  \, \Big|\, \sum_{j=0}^\infty \rho_j + \sigma_\infty = \sigma \Bigr\}.
$$
The maps $p(\mu)$, $\varphi(\sigma)$, and $	f\bigl( (\rho_j)_{j\in \N_0}, \sigma_\infty \bigr)$ are convex, moreover by Proposition~\ref{prop:pressure-legendre},
\begin{align}
	p(\mu) &= \sup_{\sigma\in [0,1]}\bigl( \mu \sigma - \varphi(\sigma)\bigr) \notag \\
		& = \sup \Biggl\{ \sum_{j=0}^\infty \mu \rho_j + \mu \sigma_\infty - f\bigl( (\rho_j)_{j\in \N_0}, \sigma_\infty \bigr)  \, \Big|\, \sum_{j=0}^\infty \rho_j + \sigma_\infty \leq 1\Biggr\}.  \label{eq:freevar}
\end{align} 
The test configuration $\rho_j \equiv 0$ and $\sigma_\infty =1$ yields $p(\mu) \geq \mu - e_\infty$ for all $\mu \in \R$, in agreement with the already known bound $p(\mu) \geq \theta^*(\mu) = \mu -e_\infty$. 
Define 
$$
	\mu_c := \inf\Bigl\{\mu \in \R \,\mid\, p(\mu) = \mu - e_\infty\Bigr\},\qquad 
	\sigma_c := \lim_{\mu \nearrow \mu_c} \frac{\dd p}{\dd \mu}(\mu).
$$
By convexity, the pressure $p$ is differentiable almost everywhere with increasing derivative, therefore $\sigma_c$ is well-defined. 

Notice $\mu_c \leq \infty$ and $\sigma_c \leq 1$. We say that the mixture of cubes undergoes a \emph{phase transition} if $\mu_c<\infty$. The phase transition is \emph{continuous} if $\sigma_c= 1$ and it is of \emph{first order} if $\sigma_c <1$, see Proposition~\ref{thm:ptgen} below.

\begin{lemma} \label{lem:hatphases}
	The following holds true: 
	\begin{enumerate} 
		\item [(a)] For each $j\in \N_0$, the map $\mu \mapsto \widehat z_j(\mu)$ is monotone increasing. 
		\item [(b)] The system undergoes a phase transition if and only if $\sum_{j=0}^\infty \widehat z_j(\mu) =\infty$ for some $\mu \in \R$, and we have 
		$$
			\mu_c = \inf \Bigl\{\mu \in \R\, \Big|\, \sum_{j\in \N_0}\widehat z_j(\mu) =\infty\Bigr\} > e_\infty.
		$$
		\item [(c)] If $\mu_c<\infty$, the phase transition is of first order if and only if $\sum_j \widehat z_j (\mu_c) <\infty$, with
		$$
			\sigma_c = 1- \prod_{j=0}^\infty \frac{1}{1+ \widehat z_j(\mu_c)}. 
		$$
	\end{enumerate} 
\end{lemma} 

\begin{proof} 
	(a) The rescaling from the proof of Lemma~\ref{lem:scale} allows us to shove the $\mu$-dependence away from the activities $z_j$ and into the vacuum activity, which becomes $\e^{-\mu}$ instead of $1$. Precisely, remembering $E(B)= E_j$ for $B\in \mathbb B_j$, we get 
	\begin{align*} 
		\Xi_{\Lambda}(\mu) & = \sum_{\{X_1,\ldots, X_n\}} \prod_{i=0}^n \e^{|X_i| \mu - E(X_i)}  
			  =  \sum_{\{X_1,\ldots, X_n\}} \e^{\mu |\cup_i X_i| - \sum_i E(X_i)} \\
			&  = \e^{\mu|\Lambda|}  \sum_{\{X_1,\ldots, X_n\}} \e^{- \mu |\Lambda \setminus \cup_i X_i|} \e^{- \sum_i E(X_i)}
	\end{align*} 
	where the sum runs over collections of pairwise disjoint cubes. Notice that $\e^{-\mu}$ appears to the power $|\Lambda\setminus \cup_i X_i|$ which is the number of vacant lattice sites. 
	We apply the equality to $\Lambda = B_{n-1}$ and find 
	\be \label{eq:zhat-rescaled}
		\widehat z_n(\mu) =\frac{z_n(\mu)}{\Xi_{B_{n-1}}(\mu)^{2^d}} =
		\e^{- E(B_n)} \times \Biggl( \sum_{\{X_1,\ldots, X_n\}} \e^{- \mu |\Lambda \setminus \cup_i X_i|} \e^{- \sum_i E(X_i)}\Biggr)^{-2^d}
	\ee
	because $\exp(\mu|B_n|) = \exp(2^d \mu|B_{n-1}|)$ cancels in the ratio defining $\widehat z_n(\mu)$. The monotonicity in $\mu$ follows. 
	
	(b) Suppose that the set $I:=\{\mu \in \R\mid \sum_{j=0}^\infty \widehat z_j (\mu) = \infty\}$ is non-empty. Then because of the monotonicity proven in (a), the set $I$ is an open or half-open interval $(\mu^*,\infty)$ or $[\mu^*,\infty)$ with $\mu^*\in \R\cup\{-\infty\}$. For  $\mu \in I$ we have $p(\mu) = \theta^*(\mu) = \mu- e_\infty$ by Theorem~\ref{thm:pressure} and~\eqref{eq:thetamu}, therefore $\mu_c \leq \mu^*< \infty$ and the system undergoes a phase transition. 
	
	It remains to check $\mu_c = \mu^*$ or equivalently, $p(\mu)>\mu-e_\infty$ for all $\mu< \mu^*$. First we show that $\mu^* >e_\infty$, which proves in particular $\mu^*>- \infty$. As noted above, $p(\mu) = \mu - e_\infty$ for all $\mu >\mu_c$. But $p(\cdot)$ is continuous because it is convex and finite, therefore the equality $p(\mu) = \mu- e_\infty$ extends to all $\mu \geq \mu^*$. On the other hand, the non-degeneracy condition $\inf_j E_j <\infty$ is enough to guarantee $p(\mu) >0$ for all $\mu\in \R$. Therefore $\mu^* - e_\infty = p(\mu^*) >0 $ and $\mu^*>e_\infty$.

	Next we show that $p(\mu)$ is continuously differentiable in $(-\infty,\mu^*)$ with derivative $\sigma(\mu) \in (0,1)$, where 
	\be \label{eq:sigmaform}
		\sigma(\mu) = 1 - \prod_{j=0}^\infty \frac{1}{1+ \widehat z_j(\mu)},
	\ee
	see Theorem~\ref{thm:density}(a). First we check  that $\sigma(\mu)$ is continuous in $(- \infty, \mu^*)$. Every effective activity $\widehat z_j(\mu)$ is a rational function of $\e^{-\mu}$ hence continuous, see~\eqref{eq:zhat-rescaled}. To deduce the continuity of $\sigma(\mu)$ we invoke dominated convergence for the series $\sum_j \log (1+ \widehat z_j(\mu))$. Fix $\mu'<\mu^*$. The monotonicity of $\widehat z_j(\mu)$ and the definition of $\mu^*$ yield $\widehat z_j(\mu) \leq \widehat z_j(\mu')$ for $(- \infty,\mu')$ with $\sum_{j=0}^\infty \log (1+ \widehat z_j(\mu'))<\infty$. Therefore dominated convergence shows $\lim_{\eps\to 0} \sigma(\mu+\eps) = \sigma(\mu)$, for all $\mu<\mu'< \mu^*$.  Thus $\sigma(\mu)$ is continuous.
	
	The differentiability of $p(\mu)$ follows from standard arguments. 
We have $p(\mu)= \lim_{n\to\infty} p_{\Lambda_n}(\mu)$ and $p'_{\Lambda_n}(\mu) = \sigma_{\Lambda_n}(\mu) \to \sigma(\mu) \in (0,1)$ by Theorem~\ref{thm:density}(a). For $\mu \in (- \infty, \mu^*)$ and $h\in \R$ small enough so that $\mu\pm h< \mu^*$, we may pass to the limit $n\to \infty$ in 
	$$
		p_{\Lambda_n}(\mu+h) - p_{\Lambda_n}(\mu) = \int_{\mu}^{\mu+h} \sigma_{\Lambda_n} (t) \dd t
	$$
	and find 
	$$
		p(\mu+h) - p(\mu) = \int_\mu^{\mu+h} \sigma(t) \dd t 
	$$
	hence $p'(\mu)= \sigma(\mu)$.
	
	The differentiability together with the inequality $\sigma(\mu) \in (0,1)$ allow us to conclude the proof of (b): write 
	 $$
	 	p(\mu^*) - p(\mu) = \int_{\mu}^{\mu^*} \sigma(u) \dd u < \mu^*- \mu
	 $$
	and 
	$$
		p(\mu)> p(\mu^*) - \mu^* + \mu = - e_\infty + \mu. 
	$$
	This holds true for all $\mu <\mu^*$, therefore $\mu_c \geq \mu^*$ and altogether $\mu_c = \mu^*> e_\infty$. 
	
	(c) As noted above, we have $p'(\mu) = \sigma(\mu)$ for all $\mu \in (-\infty,\mu^*) = (-\infty, \mu_c)$. Proceeding as in (b) but using monotone convergence for the series $\sum_j \log (1+ \widehat z_j(\mu))$ instead of dominated convergence, we obtain 
	$$
		\sigma_c = \lim_{\mu \nearrow \mu_c} p'(\mu)= \lim_{\mu \nearrow \mu_c} \sigma(\mu) = \sigma(\mu_c).
	$$	
	In particular, $\sigma_c<1$ if and only if $\sigma(\mu_c) <1$, which in turn is equivalent to $\sum_{j=0}^\infty \widehat z_j (\mu_c)<\infty$. 
\end{proof} 

\noindent In the proof of Lemma~\ref{lem:hatphases} we have proven a number of statements that can be formulated without any reference to the effective activities. 

\begin{prop} \label{thm:ptgen}
		The critical chemical potential satisfies $\mu_c > e_\infty > - \infty$. In addition:
	\begin{enumerate}
		\item [(a)] In $(-\infty, \mu_c)$ the pressure $p(\mu)$ is strictly convex and continuously differentiable  with packing fraction $p'(\mu) = \sigma(\mu) \in (0,\sigma_c)$ and it satisfies $p(\mu) > \mu - e_\infty$. 
		\item [(b)] If $\mu_c < \infty$, then $p(\mu) = \mu- e_\infty$ for all $\mu \geq \mu_c$ and the packing fraction is $\sigma(\mu)=1$. 
	\end{enumerate} 
\end{prop} 

\begin{proof}
	All statements except the strict convexity in $(- \infty, \mu_c)$ have been shown in the proof of Lemma~\ref{lem:hatphases}. The strict convexity follows from the strict monotonicity of $\sigma(\mu)$: Let $\mu_1 < \mu_2 < \mu_c$. 
Then 
$$
	\sum_{j=0}^\infty \log (1+ \widehat z_j(\mu_1)) \leq 	\sum_{j=0}^\infty \log (1+ \widehat z_j(\mu_2)) < \infty
$$
and, because of the monotonicity from Lemma~\ref{lem:hatphases}(a), 
\be \label{eq:ptdiff}
	\sum_{j=0}^\infty \Bigl( \log (1+ \widehat z_j(\mu_2)) - \log (1+ \widehat z_j(\mu_1)) \Bigr) \geq \log (1+ \widehat z_k(\mu_2)) - \log (1+ \widehat z_k(\mu_1)) 
\ee
for all $k\in \N_0$. Eq.~\eqref{eq:zhat-rescaled} shows that if $E_k<\infty$---which is the case for at least one $k\in \N_0$---then $\widehat z_k(\mu)$ is strictly increasing in $\mu$. Therefore the difference~\eqref{eq:ptdiff} is strictly positive and Eq.~\eqref{eq:sigmaform} yields  $\sigma(\mu_1) <\sigma(\mu_2)$. 
\end{proof}

\subsection{Fixed point iteration. Absence of phase transition}

The recurrence relation $\Xi_{\Lambda_{ n+1}} = z_{n+1} + (\Xi_{\Lambda_n})^{2^d}$ encountered in the proof of Theorem~\ref{thm:pressure} leads to a recurrence relation for the inverse probability of finding one large block. Indeed,
$$
	\frac{\Xi_{\Lambda_n}}{z_n} = 1+ \frac{z_{n-1}^{2^d}}{z_n} \Bigl( \frac{\Xi_{\Lambda_{n-1}}}{z_{n-1}}\Bigr)^{2^d}. 
$$
Thus if we set
$$
	v_n(\mu):= \frac{\Xi_{\Lambda_n}(\mu)} {z_{n}(\mu)}= \frac{1}{\P^\mu_{\Lambda_n}(\omega =\{\Lambda_n\})}
$$
and 
\be \label{eq:epsdef}
	\eps_n:= \frac{(z_{n-1}(\mu) )^{2^d}}{z_n(\mu)} = \exp(E_n - 2^d E_{n-1})\qquad (n\in \N),
\ee
then 
\be \label{eq:vniteration}
	v_{n} (\mu)= 1 + \eps_n \bigl(v_{n-1}(\mu)\bigr)^{2^d}\qquad (n\in \N)
\ee
and 
$$
	v_0(\mu) = 1+ \frac{1}{z_0(\mu)} = 1+ \e^{-\mu} \e^{\beta E_0}. 
$$
Notice that the $\mu$-dependence drops out from the ratio $z_{n-1}(\mu)^{2^d} / z_n(\mu)$ so that $\eps_n$ in~\eqref{eq:vniteration} does not depend on $\mu$. Thus the sequence $(v_n(\mu))_{n\in \N_0}$ is computed recursively and the only explicit $\mu$-dependence is through the initial condition $v_0(\mu)$. 

For energies $(E_n)_{n\in \N}$ leading to constant ratios $\eps_n \equiv \eps$, the iteration defining $v_n(\mu)$ is a fixed point iteration that is straightforward to analyze. Set
\be \label{eq:fepsdef} 
	f_\eps(x):= 1 + \eps x^{2^d}, \quad c_d:= \sup_{x\geq 1}\frac{x-1}{x^{2^d}}.
\ee
Notice $c_d\in (0,1)$. The following case distinction is relevant for this section and the following: 
\begin{enumerate} 
	\item If $\eps > c_d$, then $f_\eps(x)>x$ for all $x\geq 0$. 
	\item If $\eps< c_d$, then the equation $x = f_\eps(x)$ has exactly two solutions $x_-< x_+$ in $(0,\infty)$. They satisfy $1 \leq x_- < x_+$. The smaller fixed point is attractive ($f'_\eps(x_-)\in (0,1)$), the larger fixed point is repulsive ($f'_\eps(x_+)>1$).
	\item If $\eps = c_d$, then $f_\eps$ has exactly one fixed point. The fixed point satisfies $f'_\eps(x) =1$. 
\end{enumerate} 

\begin{theorem} \label{thm:ptabsence}
	Suppose 
	$$
	\liminf_{j\to \infty} \eps_j = \liminf_{j\to \infty} \exp( E_j -  2^d E_{j-1}) > c_d.
	$$
	Then $\mu_c =\infty$. 
\end{theorem} 

\noindent Because of $c_d< 1$, the theorem applies in particular to the reference measure for which $E_j\equiv 0$ and we find that there are no entropy-driven phase transitions. 

\begin{cor}
	If $E_j\equiv 0$, then $\mu_c = \infty$. 
\end{cor} 

\begin{proof}[Proof of Theorem~\ref{thm:ptabsence}]
	Fix $\mu \in \R$ and suppress the $\mu$-dependence from the notation. By the assumption of the theorem there exists $n_0\in \N$
	and $\eps>c_d$  such that $\eps_n > \eps$ for all $n\geq n_0$. Then $v_{n_0+k} \geq f_\eps^k (v_{n_0})$ for all $k\in \N_0$. A close look at the fixed point iteration $x_{k+1} = f_\eps(x_k)$, based on the case distinction sketched above, shows that $f_\eps^k(x_0)$ goes to infinity for all $x_0\geq 0$. Consequently $v_n\to \infty$ as $n\to \infty$. 
	We check that the divergence is in fact exponentially fast. 
	For $n \geq n_0$ we have $v_n =1+ \eps_n v^{2^d}_{n-1}\geq \eps v_{n-1}^{2^d}$ hence for all $\delta >0$, 
	$$
		\delta v_n \geq \delta^{1- 2^d} \eps \times (\delta v_{n-1})^{2^d}.
	$$
	Let $\delta>0$ be the solution of $\delta^{1- 2^d} \eps =1$, then 
	$$
		\frac{1}{|B_n|} \log(\delta v_n) \geq 	\frac{1}{|B_{n-1}|} \log(\delta v_{n-1})
	$$
	for all $n\geq n_0$. Pick $k\geq n_0$ with $\delta v_{k}> 1$, which exists because of $v_n\to \infty$. Then for all $n\geq k$ we have 
	$$
		\delta v_n \geq (\delta v_k)^{|B_n|/|B_k|}.
	$$
	In particular $v_n\to \infty$ exponentially fast. 
	To conclude, we turn back to the pressure, bring the $\mu$-dependence back into the notation, and note
	$$
		p(\mu) - (\mu - e_\infty) = \liminf_{n\to \infty} \frac{1}{|B_n|} \log \frac{\Xi_{\Lambda_n}(\mu)}{z_n(\mu)} = \liminf_{n\to \infty}\frac{1}{|B_n|}\log v_n(\mu) >0.
	$$
	Thus $p(\mu)> \mu - e_\infty$. This holds true for every $\mu \in \R$, therefore $\mu_c =\infty$. 
\end{proof} 

\subsection{Continuous phase transition. Scaling limit} 

Here we consider a model where each block has the same energy. Thus we assume that for some $\lambda\in\R$, 
$$
	\forall j \in \N_0:\quad E_j = \lambda. 
$$
The total energy $\sum_{B\in \omega} E(B)$ is then simply $\lambda$ times the number of blocks in a configuration, the Boltzmann factor is given by $\e^{-\lambda}$ to the power of the number of blocks, a feature somewhat  reminiscent of random cluster models \cite[Chapter 6]{georgii-haggstrom-maes2001}.

The constant sequence $E_j\equiv \lambda$ has $e_\infty = \lim_{j\to \infty}E_j/|B_j| =0$. The ratio $\eps_n$ from~Eq.~\eqref{eq:epsdef} is constant and equal to
$$
	\eps(\lambda):= \e^{- (2^d - 1)\lambda}. 
$$
We can therefore analyze the system with the fixed point iteration from the previous section. Set 
$$
	\lambda_d:= - \frac{\log c_d}{2^d-1}
$$
and notice $\lambda_d>0$. 
If $\eps(\lambda)>c_d$ i.e. $\lambda< \lambda_d$, then Theorem~\ref{thm:ptabsence} tells us that $\mu_c=\infty$ and the system has 
no phase transition. 

If $\eps(\lambda) <c_d$ i.e. $\lambda > - (2^d-1)^{-1} \log c_d$, then by case (2) below~\eqref{eq:fepsdef}, the function $f_{\eps(\lambda)}(x)$ has two fixed points $0< x_-(\lambda)< x_+(\lambda)$. 

\begin{theorem} \label{thm:ptcontinuous}
	Assume $\lambda > \lambda_d= - (2^d -1)^{-1}\log c_d$ and let $x_+(\lambda)>1$ be the repulsive fixed point of the map $\R_+\ni x\mapsto 1+ \eps(\lambda) x^{2^d}$.  Then the system undergoes a phase transition at 
	$$
		\mu_c(\lambda) = \lambda - \log \bigl( x_+(\lambda) - 1\bigr)
	$$
and the phase transition is continuous.
\end{theorem}

\begin{proof}
	To lighten notation we suppress the $\lambda$-dependence. Set $\mu^*:= \lambda - \log (x_+ - 1)$ and note 
	$$
		v_0(\mu^*) = 1+ \exp(- \mu^*+\lambda) =x_+(\lambda).
	$$	
	 Our task is to show $\mu_c = \mu^*$. To that aim we return to the fixed point iteration for the inverse probability of finding a large block and the case distinction below~\eqref{eq:fepsdef}: 
	\begin{enumerate} 
		\item If $\mu>\mu^*$, then $v_0(\mu)<x_+(\lambda)$ and $v_0(\mu)$ belongs to the domain of attraction of the fixed point $x_-(\lambda)$ and $v_n(\mu)\to x_-(\lambda)$ as $n\to \infty$. 
		\item If $\mu = \mu^*$, then $v_0(\mu) = x_+(\lambda)$ and $v_n(\mu) = x_+(\lambda)$ for all $n\in \N_0$. 
		\item If $\mu< \mu^*$, then $v_0(\mu) > x_+(\lambda)$ and $v_n(\mu) \to \infty$.
\end{enumerate} 
	In the cases (1) and (2) we have 
	$$
			p(\mu) - \mu = \lim_{n\to \infty} \frac{1}{|B_n|}\log \frac{\Xi_{\Lambda_n}(\mu)}{z_n(\mu)} = \lim_{n\to \infty} \frac{1}{|B_n|}\log v_n(\mu) =0.
	$$
	Thus $p(\mu) = \mu$ for all $\mu \geq \mu^*$. 
	Proceeding as in the proof of Theorem~\ref{thm:ptabsence}, one shows that the divergence in case (3) is exponentially fast and concludes $p(\mu)>\mu$. 
	Thus $p(\mu) = \mu$ if and only if $\mu\geq \mu^*$, consequently $\mu_c = \mu^* <\infty$. In particular, the system undergoes a phase transition. 
		
	The effective activity at $\mu = \mu_c$ is given by 
	$$
		\widehat z_j(\mu_c) = \exp\Bigl(- \lambda + |B_j|\bigl(\mu_c - p_{j-1}(\mu_c)\bigr)\Bigr).
	$$
	Because of $\mu_c = p(\mu_c) \geq p_{j-1}(\mu_c)$, it follows that 
	$\widehat z_j(\mu_c) \geq \exp( - \lambda)$ and $\sum_{j=0}^\infty \widehat z_j(\mu_c) =\infty$.  We deduce from Lemma~\ref{lem:hatphases}(c) 
	that the phase transition is continuous. 
\end{proof} 

\noindent The mixture of hierarchical cubes is closely related to Mandelbrot's percolation process~\cite{mandelbrot1982fractal, chayes-chayes-durrett88}. 
Let us define a sequence of random subsets of the unit cube by rescaling $\Lambda_n=\{1,\ldots, 2^n\}^d$. Let $\mathcal K$ be the collection of compact subsets of $[0,1]^d$, equipped with the Hausdorff distance and Borel $\sigma$-algebra $\mathcal B_\mathcal K$. Let us first map a block $B\subset \mathbb B\subset \Z^d$ to its continuum counterpart $B'\subset \R^d$ given by 
$$
	B' = \bigcup_{\vect k \in B} \bigl[ k_1-1,k_1]\times \cdots \times \bigl[ k_d-1,k_d].
$$
Thus $B'$ is the cube in $\R^d$ obtained as the union of unit cubes with upper right corners $\vect k \in B\subset \Z^d$.
If $B\subset \Lambda_n$ then $B'\subset[0, 2^n]^d$. 
For $n\in \N_0$, define the random variable $K_n:(\Omega_{\Lambda_n},\mathcal P(\Omega_{\Lambda_n}), \mathbb P_{\Lambda_n})\to (\mathcal K,\mathcal B_\mathcal K)$ by 
$$
	K_n(\omega):= \bigcup_{B\in \omega} \frac1{2^n} B'.
$$
Further let $F_n(\omega)$ be the closure of $[0,1]^d \setminus K_n(\omega)$. 
The random set $K_n(\omega)$ is constructed as a union of cubes of sidelengths $1, \frac 12,\ldots, \frac1{2^n}$,  roughly as follows. 
\begin{itemize}
	\item With probability $1/v_n(\mu)$ the random set is equal to the whole unit cube, $K_n(\omega) = [0,1]^d$. 
	\item With probability $1 - 1/v_n(\mu)$, the random set is strictly smaller than the whole unit cube. In that case we decide independently for each  of the $2^d$ subcubes ($[0,\frac12]^d$ and its translates) whether to add or not add it to $K_n(\omega)$; a subcube is added with probability $1/v_{n-1}(\mu)$. This results in a set $A_{n,1}(\omega)$ that is a union of cubes of sidelength $1/2$. Then, for each subcube that has not been added, we repeat the construction for each of the $2^d$ subsubcubes, to be added with probability $1/v_{n-2}(\mu) $. We iterate until we have reached the smallest cubes of sidelength $2^{-n}$, associated with the probability $1/v_0(\mu)$. 	
\end{itemize} 
If the sequence $v_n(\mu)$ is $n$-independent, let us write $q\equiv 1/v_n(\mu)$, $p= 1- q$, and suppress the $\mu$-dependence. Then we may think of $K_n$ as a growing family of subsets of $[0,1]^d$ and accordingly of $F_n(\omega)$ as a decreasing family, and set $F(\omega) = \cap_{n\in \N_0} F_n(\omega)$; we owe to S. Winter the remark that $F(\omega)$ should correspond to a special instance of Mandelbrot's percolation process~\cite{mandelbrot1982fractal, chayes-chayes-durrett88}. 

Revisiting the case distinctions on the asymptotic behavior of $(v_n(\mu))_{n\in \N_0}$ we may expect the following behavior, under the assumption $\lambda>\lambda_d$ and after restoration of the $\mu$-dependence in the notation:
\begin{enumerate} 
	\item If $\mu = \mu_c(\lambda)$ then as $n\to \infty$ the distribution of $K_n^\mu$ should converge in some suitable sense to a process where at each scale, a block is added with probability $1/x_-(\lambda)$, with $x_-(\lambda)$ the repulsive fixed point of $x\mapsto 1 + \eps(\lambda) x^{2^d}$. 
	\item If $\mu >\mu_c(\lambda)$ the distribution of $K_n^\mu$ should converge in some suitable sense to a process where at each scale, a block is added with probability $1/x_+(\lambda)$, with $x_+(\lambda)$ the attractive fixed point of $x\mapsto 1 + \eps(\lambda) x^{2^d}$. 
\end{enumerate} 
A rigorous statement and proof (or disproof) of these statements are beyond the scope of this article.

\subsection{First-order phase transition} \label{sec:firstorder}

Finally  we provide necessary and sufficient conditions for the existence of a first-order phase transitions. The mathematical proofs carried out in this section are complemented by a heuristic discussion in Section~\ref{sec:discussion}. 

\begin{theorem} \label{thm:firstorder}
	Set $u_j:= \exp( |B_j| e_\infty - E_j)$. 
	The following two conditions are equivalent: 
	\begin{enumerate} 
		\item [(i)] There exists a family of non-negative weights $(a_k)_{k\in \N_0}$ such that $\sum_{j=0}^\infty u_j \exp(a_j)<\infty$ and 
	\be \label{eq:suff1}
		\sum_{k=j}^\infty \frac{|B_j|}{|B_k|} \log \bigl( 1+ u_k\, \e^{a_k}\bigr) \leq a_j
	\ee
	for all $j\in \N_0$.
		\item [(ii)] The mixture of cubes has a first-order phase transition. 
	\end{enumerate} 
\end{theorem} 

\begin{cor} \label{cor:firstorder} \hfill 
	\begin{enumerate} 
		\item [(a)] If there is a first-order phase transition, then necessarily $E_j \geq |B_j| e_\infty$ (i.e., $u_j\leq 1$) for all $j\in \N_0$ and $\sum_{j=0}^\infty u_j< \infty$. 
		\item [(b)] The condition $\sum_{j=0}^\infty u_j \leq 1/\e$ is sufficient for the existence of a first-order phase transition. 
	\end{enumerate} 
\end{cor} 

\begin{example}
	Let $E_j = J ( - |B_j| + |\partial B_j|)$ with $J>0$ some coupling constant and $|\partial B_j| = 2d\, 2^{j(d-1)}$ the area of the boundary of a cube of sidelength $2^j$ in $\R^d$. Then if $d\geq 2$ and $J$ is sufficiently large, the mixture of cubes has a first-order phase transition. 
\end{example}

\begin{proof}[Proof of Corollary~\ref{cor:firstorder}]
	(a) If there is a first-order phase transition, then by condition (i) in  Theorem~\ref{thm:firstorder} we must have $\sum_{j=0}^\infty u_j \leq \sum_{j=0}^\infty u_j\exp(a_j)<\infty$, moreover 
	$\log (1+ u_j \exp(a_j)) \leq a_j$ hence $u_j \leq 1- \exp( - a_j)\leq 1$. 
	
	(b) Choose $a_k\equiv 1$. Because of $\log(1+x) \leq x$ and $|B_j|\leq |B_k|$ whenever $j\leq k$ we have
	$$
		\sum_{k=j}^\infty \frac{|B_j|}{|B_k|} \log \Bigl(1+ u_k\e^{a_k}\Bigr) \leq \sum_{k=0}^\infty u_k \e^{a_k} = \Bigl(\sum_{k=0}^\infty u_k\Bigr) \e \leq 1 = a_j.
	$$
	Thus condition (i) in Theorem~\ref{thm:firstorder} is satisfied and the mixture has a first-order phase transition. 
\end{proof}

\begin{proof}[Proof of the implication $(ii)\Rightarrow (i)$ in Theorem~\ref{thm:firstorder}]
	
Suppose that the mixture of cubes has a first-order phase transition. Then 
$$
	\mu_c - e_\infty =p(\mu_c) =  \sum_{j=0}^\infty \frac{1}{|B_j|}\log (1+ \widehat z_j(\mu_c))
$$
hence 
\begin{align*}
	\widehat z_j(\mu_c) &= \exp \Bigl( |B_j|\mu_c - E_j\Bigr) \exp\Bigl( p(\mu_c) - |B_j|\,p_{j-1}(\mu_c)\Bigr) \\
		& = \exp\bigl( |B_j| e_\infty - E_j\bigr) \,\exp\Bigl(|B_j| \sum_{k=j}^\infty \frac{1}{|B_k|}\log (1+ \widehat z_k(\mu_c))\Bigr)
\end{align*} 
for all $j \in \N_0$. Equivalently,  $\zeta_j := \widehat z_j(\mu_c)$ and $u_j := \exp( |B_j| e_\infty - E_j)$, satisfy
\be \label{eq:fp1} 
	\zeta_j  = u_j \exp\Bigl(|B_j| \sum_{k=j}^\infty \frac{1}{|B_k|}\log (1+ \zeta_k)\Bigr) \quad (j\in \N_0). 
\ee
Define $a_j: =\log(\zeta_j/u_j)$, then $a_j\geq 0$ and the inequality~\eqref{eq:suff1} holds true and is actually an equality. Moreover
$$
	\sum_{j=0}^\infty u_j \e^{a_j} = \sum_{j=0}^\infty \zeta_j = \sum_{j=0}^\infty\widehat z_j(\mu_c) <\infty
$$
because the phase transition is of first order, see Lemma~\ref{lem:hatphases}(b).
\end{proof} 

\noindent The strategy for the proof of the implication $(i) \Rightarrow (ii)$ in Theorem~\ref{thm:firstorder} is as follows. First we show that if condition (i) holds true, then the fixed point equation~\eqref{eq:fp1} has at least one solution $(\zeta_j)$, see Lemma~\ref{lem:fp}. Then we turn to the computation of the free energy $\varphi(\sigma)$, which is given by a constrained minimization; we show that every solution of the fixed point problem~\eqref{eq:fp1} is associated with a critical point of the Lagrange functional $L(\vect \rho,\sigma_\infty,\mu)$ and deduce that the free energy is affine on some interval $[\sigma^*,1]$.

\begin{lemma} \label{lem:fp}
If the inequality~\eqref{eq:suff1} holds true for some family of non-negative weights $(a_k)_{k\in \N_0}$, then the fixed point problem~\eqref{eq:fp1} has at least one solution $\vect \zeta \in \R_+^\N$ that satisfies $\zeta_j \leq u_j \exp( a_j)$ for all $j\in \N_0$. 
\end{lemma} 

\begin{proof} 
	We adapt the treatment of tree fixed points by Faris~\cite[Section 3.1]{faris2010combinatorics} and reformulate our problem as a fixed point problem in a partially ordered set for a monotone increasing map. 
	Let $\mathcal L$ be the space of bounded non-negative sequences $\vect z = (\zeta_j)_{j\in \N_0}$. For $\vect \zeta\in \mathcal L$, define
	$$
		F_j(\vect \zeta) := u_j \exp\Bigl(  \sum_{k=j}^\infty \frac{ |B_j|}{|B_k|}\log \bigl(1 + \zeta_k\bigr)\Bigr)\quad (j\in \N_0).
	$$
	Further set $\vect F(\vect \zeta):= (F_j(\zeta))_{j\in \N_0}$. If $(u_j)_{j\in \N_0}$ is bounded, then $\vect F(\vect \zeta)$ is bounded as well; thus $\vect F$ maps $\mathcal L$ to $\mathcal L$. 
We equip $\mathcal L$ with the partial order of pointwise inequality, i.e., $\vect x \leq \vect y$ if and only if $x_j \leq y_j$ for all $j \in \N_0$, and note that  $\vect F$ is increasing with respect to that partial order. 
	
The vector $\vect w$ defined by $w_k:= u_k\exp(a_k)$ satisfies $F_k(\vect w) \leq w_k$ for all $k\in \N_0$. Define a sequence $({\vect \zeta}^{(n)})_{n\in \N_0}$ iteratively by $\zeta_j^{(0)}\equiv 0$ and $\zeta_j^{(n+1)} = F_j(\vect \zeta^{(n)})$. Notice $\zeta_j^{(1)} = u_j$.
		
	We check by induction over $n$ that $\zeta_j^{(n)}\leq \zeta_{j}^{(n+1)}\leq u_j \exp(a_j)=w_j$ for all $j \in \N_0$ and $n\in \N_0$. 
	For $n=0$, the inequality reads $0 \leq u_j \leq w_j$ which is clearly true.
	The induction step works because of the monotonicity of $\vect F$ and because of $\vect F(\vect w) \leq \vect w$.

	It follows that the limit $\zeta_j:= \lim_{n\to \infty}\zeta_j^{(n)}$ exists for all $j\in \N_0$ and satisfies $\zeta_j \leq w_j$, moreover $\vect \zeta = \vect F(\vect \zeta)$ because $F_j(\vect \zeta^{(n)})\to F_j(\vect \zeta)$ by monotone convergence.
\end{proof}

\noindent The solution of Lemma~\ref{lem:fp} is in fact a critical point of the Lagrange function for  the computation of the free energy $\varphi(\sigma)$. Let 
\be
	L_\sigma(\vect \rho,\sigma_\infty;\mu):= f\bigl( \vect \rho, \sigma_\infty\bigr) - \mu\Bigl( \sum_{j=0}^\infty \rho_j + \sigma_\infty - \sigma\Bigr). 
\ee
Given $(\zeta_j)_{j\in \N_0}\in\R_+^{\N_0}$ a summable sequence, set 
\be \label{eq:stardef}
		\mu^* := e_\infty + \sum_{k=0}^\infty \frac{1}{|B_k|} \log (1+ \zeta_k),\quad \rho^*_j = \frac{\zeta_j}{1+\zeta_j}\prod_{k=j+1}^\infty \frac1{1+\zeta_k},\quad 	\sigma^*:= 1- \prod_{j=0}^\infty \frac{1}{1+\zeta_j}.
\ee
Note $\sigma^*=\sum_{j=0}^\infty \rho^*_j \in (0,1)$. Fix $\sigma \in [\sigma^*,1)$ and define 
\be \label{eq:euler-lagrange1}
	\sigma_\infty:= \frac{\sigma - \sigma^*}{1 -\sigma^*},\quad 
	\rho_j := (1-\sigma_\infty) \rho^*_j.
\ee
Thus $(\vect \rho,\sigma_\infty) $ is a convex combination
\be \label{eq:minco}
	(\vect \rho,\sigma_\infty)  =  (1-\sigma_\infty)\,  (\vect \rho^*, 0) +\sigma_\infty (\vect 0,1)
\ee
and the packing fraction $\sigma$ enters only via the weight $\sigma_\infty$ in the convex combination. 

\begin{lemma}  \label{lem:euler-lagrange1}
	Suppose that the system~\eqref{eq:fp1} admits a solution $\vect \zeta\in \R_+^{\N_0}$ that satisfies $\sum_{j=0}^\infty \zeta_j <\infty$ and define $\mu^*, \sigma^*,\vect{\rho^*}$ as in~\eqref{eq:stardef}. Assume  $\sigma \in [\sigma^*,1)$ and define $(\vect \rho,\sigma_\infty)$  by~\eqref{eq:euler-lagrange1}. Then all partial derivatives of $L$ at $(\vect \rho,\sigma_\infty, \mu^*)$ exist and are equal to zero, and $(\vect \rho,\sigma_\infty, \mu^*)$ is a minimizer of the Lagrange functional $L$. 
\end{lemma} 

\begin{proof} 
	Remember 
	$$
		\widehat \rho_j = \frac{\rho_j}{1 - \sum_{k \geq j+1} \rho_k- \sigma_\infty} = \frac{\rho'_j}{1- \sum_{k\geq j+1} \rho'_k},\quad \rho'_j = \frac{\rho_j}{1-\sigma_\infty}. 
	$$
	Lemma~\ref{lem:equiv} applied to $m=0$ and $(\rho'_j)$ and $(\zeta_j)$ yields 
	\be \label{eq:rhozeta}
		\widehat \rho_j  = \frac{\zeta_j}{1+ \zeta_j}<1 \qquad (j\in\N_0).
	\ee
	The convergence of the series $\sum_j \zeta_j$ implies $\zeta_j \to 0$ as $j\to \infty$. 
%
	The free energy is given by a linear term minus the entropy, and the partial derivatives of the entropy have been computed in Eqs.~\eqref{eq:chem1}--~\eqref{eq:chem2}. The existence of the partial derivatives follows from Proposition~\ref{prop:analyticity} and $\rho_j>0$ for all $j$. We obtain 
	\begin{align}
		\frac{\partial L_\sigma}{\partial \rho_j}(\vect \rho,\sigma_\infty,\mu^*) & = 
			\frac{1}{|B_j|}\Biggl( E_j + \log \frac{\widehat \rho_j}{1-\widehat \rho_j} - \sum_{k=0}^{j-1} \frac{|B_j|}{|B_k|} \log (1- \widehat \rho_k) - \mu^* |B_j|\Biggr) \label{eq:L1} \\ 
		\frac{\partial L_\sigma}{\partial \sigma_\infty}(\vect \rho,\sigma_\infty,\mu^*) & = 
		e_\infty - \sum_{j=0}^\infty \frac{1}{|B_j|} \log (1- \widehat \rho_j) - \mu^* \label{eq:L2}.
	\end{align} 
	Eq.~\eqref{eq:rhozeta} yields $\log (1+ \zeta_j) = - \log (1- \widehat \rho_j)$. Eq.~\eqref{eq:L2} then follows from the definition of $\mu^*$ in~\eqref{eq:euler-lagrange1} and Eq.~\eqref{eq:L1} follows from~\eqref{eq:euler-lagrange1} and~\eqref{eq:L2}. Finally we note
$$
	\frac{\partial L_\sigma}{\partial \mu^*}(\vect \rho,\sigma_\infty,\mu^*) =  (1- \sigma_\infty) \sigma^* + \sigma_\infty = \sigma 
$$
	by definition of $\sigma_\infty$. 
	
		By convexity, the critical point is a minimizer in every finite-dimensional affine subspace obtained by changing only finitely many components of $(\vect \rho,\sigma^*_\infty,\mu^*)$. The union of these subspaces in dense, and the Lagrange functional is continuous in the domain $||(\vect \rho,\sigma_\infty)||\leq 1$; the lemma follows. 
\end{proof} 

\begin{lemma} \label{lem:euler-lagrange2}
	For $\sigma\in [\sigma^*,1)$ the vector $(\vect \rho, \sigma_\infty,\mu^*)$ defined in~\eqref{eq:euler-lagrange1} is a minimizer of the free energy $f(\vect \rho,\sigma)$ under the constraint $\sum_{j=0}^\infty \rho_j + \sigma_\infty = \sigma$,  and the minimum $\varphi(\sigma)$ is an affine function of $\sigma$ with slope $\mu^*$, 
	$$
		\varphi(\sigma) = \varphi(\sigma^*) + \mu^*(\sigma- \sigma^*) \qquad (\sigma^*\leq \sigma <1). 
	$$
\end{lemma} 

\begin{proof}
	The vector $(\vect \rho,\sigma_\infty)$ is a minimizer because of Lemma~\ref{lem:euler-lagrange1}. By~\eqref{eq:minco} and Lemma~\ref{lem:convex-combi}, the free energy is 
	\begin{align*}
		\varphi(\sigma) & = f( \vect \rho,\sigma_\infty) = (1-\sigma_\infty) f(\vect {\rho^*},0) + \sigma_\infty f(\vect 0,1)
			= (1-\sigma_\infty) \varphi(\sigma^*) + \sigma_\infty e_\infty.
	\end{align*}
	Since $\sigma_\infty $ is an affine function of $\sigma$ by~\eqref{eq:euler-lagrange1} it follows that $\varphi(\sigma)$ is an affine function of $\sigma$ as well. Lemma~\ref{lem:euler-lagrange1} yields 
	$$
		\frac{\partial f}{\partial \rho_j}(\vect{\rho^*},0) = 		\frac{\partial f}{\partial \sigma_\infty}(\vect{\rho^*},0) = \mu^*.
	$$
	Therefore
	\begin{align*} 
		\varphi'(\sigma)  & = \sum_{j=0}^\infty \frac{\partial f}{\partial \rho_j}(\vect{\rho^*},0) \frac{\partial \rho_j}{\partial \sigma} + \frac{\partial f}{\partial \sigma_\infty}(\vect{\rho^*},0) \frac{\partial \sigma_\infty}{\partial \sigma} \\
			& = \sum_{j=0}^\infty \mu^* \Bigl( - \frac{\rho^*}{1-\sigma^*} \Bigr)+ \frac{\mu^*}{1-\sigma^*} = \mu^*. \qedhere
	\end{align*} 
\end{proof} 

\begin{lemma}  \label{lem:starcrit}
	We have $\mu^*=\mu_c$, $\sigma^*=\sigma_c$, and $\zeta_j = \widehat z_j(\mu_c)$ for all $j\in \N_0$. 
\end{lemma} 

\begin{remark}
	It follows that the solution $\vect \zeta$ of the fixed point problem~\eqref{eq:fp1} is in fact unique. 
\end{remark} 

\begin{proof} 
 It follows from Lemma~\ref{lem:euler-lagrange2} and elementary considerations on Legendre transforms that $p(\mu) = \sup_{\sigma\in [0,1]} (\mu\sigma - \varphi(\sigma)) = \mu- e_\infty$ for $\mu \geq \mu^*$, which yields $\mu_c \leq\mu^*$. 
 
 Moreover, for $\mu>\mu^*$ the unique maximizer of $\sigma \mapsto \mu \sigma - \varphi(\sigma)$ is $\sigma =1$ while for $\mu = \mu^*$ every $\sigma \in [\sigma^*,1]$ is a maximizer. In particular, $p(\mu^*) = \sigma^* \mu^* - \varphi(\sigma^*)$ and the constrained minimizer $(\vect{\rho^*},0)$ of $f(\vect \rho,\sigma_\infty)$ is a maximizer at $\mu = \mu^*$ in the variational formula~\eqref{eq:freevar} for the pressure. 
	It follows from Proposition~\ref{prop:pressure-legendre} that $\sum_{j=0}^\infty \widehat z_j(\mu^*)<\infty$---otherwise, the unique maximizer would be $(\vect 0,1)$, in contradiction with $(\vect{\rho^*},0)$ be a maximizer---hence by Lemma~\ref{lem:hatphases}, we must have $\mu^*\leq \mu_c$. 
	
	Thus we have shown $\mu_c = \mu^*<\infty$.	
%
	Proposition~\ref{prop:pressure-legendre} and the previous considerations on the variational formula for the pressure $p(\mu^*) = p(\mu_c)$ also yield 
	$$
		\widehat {\rho_j^*} = \frac{\widehat z_j(\mu_c)}{1+\widehat z_j(\mu_c)} = \frac{\zeta_j}{1+\zeta_j}
	$$
	hence $\zeta_j = \widehat z_j(\mu_c)$ for all $j\in \N_0$. 
	Finally $\sigma_c = \sum_{j=0}^\infty \rho_j^* =\sigma^*$. 
\end{proof} 

\begin{proof}[Proof of the implication $(i)\Rightarrow (ii)$ in Theorem~\ref{thm:firstorder}]
	Suppose that condition (i) is satisfied. Then by Lemma~\ref{lem:fp} the fixed point equation~\eqref{eq:fp1} has a solution and we may define $\mu^*\in \R$, 
	$\sigma^*\in (0,1)$, and $\rho^*_j$ as in~\eqref{eq:stardef}. Lemma~\ref{lem:starcrit} shows that the system has a phase transition at $\mu_c = \mu^*$ with $\sigma_c = \sigma^*<1$, hence the transition is of first order. 
\end{proof}

\section{Discussion} \label{sec:discussion}

A concluding heuristic discussion of the parameter-dependent model from Section~\ref{sec:phasetransition} makes the connection to the motivating  considerations on the mixture of hard spheres in the introduction more apparent. By Proposition~\ref{prop:analyticity}, the free energy~\eqref{eq:freenergy}  of the parameter-dependent model is 
$$
	f(\vect \rho,\sigma_\infty) = \sum_{j=0}^\infty \rho_j \frac{E_j}{|B_j|} + \sigma_\infty e_\infty + \sum_{j=0}^\infty \rho_j \bigl(\log \rho_j - 1\bigr)  + \Phi(\vect \rho,\sigma_\infty)
$$
with $\Phi(\vect \rho,\sigma_\infty)$ the absolutely convergent power series from Eq.~\eqref{eq:taylor}. The leading order in the power series is quadratic,
$$
	\Phi(\vect \rho,\sigma_\infty) = \frac12 \sum_{j=0}^\infty \frac{\rho_j}{|B_j|}\Bigl( \rho_j + 2\sum_{k=j+1}^\infty \rho_k + 2\sigma_\infty\Bigr) + \text{higher order terms}
$$
and the power series vanishes when $\rho_j \equiv 0$. Every configuration is a convex combination of a gas configuration and a condensed configuration 
$$
	(\vect \rho,\sigma_\infty) = (1- \sigma_\infty) \, (\vect{\rho'},0) + \sigma_\infty \, (0,1)
$$
and by Lemma~\ref{lem:convex-combi} the free energy is 
$$
	f(\vect \rho,\sigma_\infty) = (1- \sigma_\infty) f(\vect{\rho'},0) + \sigma_\infty e_\infty,
$$
which implies 
\be \label{eq:phiscale}
	\Phi(\vect \rho,\sigma_\infty) = - \sum_{j=0}^\infty \frac{\rho_j}{|B_j|} \log (1- \sigma_\infty) + (1-\sigma_\infty) \Phi(\vect{\rho'},0).
\ee
When minimizing the free energy at prescribed packing fraction $\sigma_\infty + \sum_{j=0}^\infty \rho_j = \sigma$ two scenarios are possible: In the gas phase the minimizer has $\sigma_\infty =0$ while in the coexistence region the minimizer has $\sigma_\infty \in (0,1)$. Accordingly in the gas phase the minimizer solves 
$$
	\frac{E_j}{|B_j|} + \frac{1}{|B_j|}  \log \rho_j + \frac{\partial \Phi}{\partial\rho_j}(\vect \rho,0) = \mu \qquad (j \in \N_0)
$$
with $\mu \in\R$ some Lagrange parameter determined by 
$$
	\sum_{j=0}^\infty \rho_j = \sum_{j=0}^\infty \exp\Bigl( \mu|B_j| - E_j -  |B_j| \frac{\partial \Phi}{\partial\rho_j}(\vect \rho,0)  \Bigr) =\sigma.
$$
In the coexistence region the equations are instead 
\begin{align*}
	\frac{E_j}{|B_j|} + \frac{1}{|B_j|}  \log \rho_j + \frac{\partial \Phi}{\partial\rho_j}(\vect \rho,\sigma_\infty) &= \mu \qquad (j \in \N_0),\\
	e_\infty + \frac{\partial \Phi}{\partial \sigma_\infty} (\vect \rho,\sigma_\infty)&= \mu,\\
	\sigma_\infty + \sum_{j=0}^\infty \rho_j & = \sigma.
\end{align*}
The second equation allows us to eliminate the Lagrange multiplier $\mu$ from the first equation, we obtain 
\be \label{eq:ptlagrange}
	\rho_j \exp\Biggl(|B_j|\Bigl( \frac{\partial \Phi}{\partial\sigma_\infty}(\vect \rho,\sigma_\infty) - \frac{\partial \Phi}{\partial\rho_j}(\vect \rho,\sigma_\infty) \Bigr) \Biggr) = \exp\bigl( |B_j| e_\infty - E_j\bigr) \qquad (j\in \N_0). 
\ee
Equation~\eqref{eq:phiscale} allows us to formulate instead equations in terms of primed variables $\rho'_j = \rho_j /(1-\sigma_\infty)$. 
Indeed, 
\begin{align*}
	\frac{\partial \Phi}{\partial \rho_j}(\vect \rho,\sigma_\infty) &=  - \frac{1}{|B_j|} \log (1- \sigma_\infty) + \frac{\partial \Phi}{\partial \rho_j}(\vect{ \rho'},0)\\
	\frac{\partial \Phi}{\partial \sigma_\infty}(\rho,\sigma_\infty) & = -  \sum_{j=0}^\infty \frac{\rho'_j}{|B_j|} - \Phi(\vect{\rho'},0) + \sum_{j=0}^\infty \rho'_j \frac{\partial \Phi}{\partial \rho_j}(\vect{ \rho'},0)
\end{align*} 
and~\eqref{eq:ptlagrange} is of the form 
\be \label{eq:lagrange-final}
	\rho'_j \exp\bigl( F_j(\vect{ \rho'})\bigr) = u_j \qquad (j\in \N_0)
\ee
with $u_j =  \exp( |B_j| e_\infty - E_j)$ and $F_j(\vect \rho')$ a power series that is absolutely convergent in $||\vect{\rho'}|| = \sum_{j=0}^\infty |\rho'_j|<1$ and satisfies $F_j(\vect{\rho'}) =O( ||\vect{\rho'}||)$. 
The fixed point equation~\eqref{eq:lagrange-final} is similar to~\eqref{eq:fp1}. In the absence of the correction term $F_j$ the solution would be $\rho'_j = u_j$. For sufficiently small values of $u_j$ the solution should be a power series in the variables $u_j$. Rigorous statements can be derived with the inversion theorems from~\cite{jttu2014,jansen-kuna-tsagkaro2019}, complementing Lemma~\ref{lem:fp} on the solvability of Eq.~\eqref{eq:fp1}. 

\subsubsection*{Acknowledgments}
	I thank Serena Cenatiempo, Diana Conache, Dimitrios Tsagkarogiannis, and Steffen Winter for stimulating discussions, and David Brydges for pointing out the possible usefulness of renormalization for treating mixtures of objects of different sizes. 
	
\medskip

\bibliographystyle{amsalpha}
\bibliography{hierarchical}

\end{document}